\documentclass[a4paper,UKenglish,cleveref,autoref,thm-restate,numberwithinsect]{lipics-v2021}%
\usepackage[utf8]{inputenc}
\usepackage{amsmath,amsthm,amssymb,tabularx,xspace,framed,enumitem,physics,multirow,mdframed,tikz}
\usetikzlibrary{decorations.pathreplacing,calligraphy,patterns.meta, patterns,arrows,positioning}

\newcommand{\eps}{\varepsilon}
\newtheorem{invariant}[theorem]{Invariant}

\newtheorem{goodsit}{Good Situation}

\renewcommand\S{\mathcal{S}}

\newcommand\nice{\mathcal{N}}
\newcommand\wcomplete{\mathcal{L}}
\newcommand\relev{\mathcal{R}}

\newcommand\calC{{\mathcal C}}
\newcommand\calD{{\mathcal D}}
\newcommand\calE{{\mathcal E}}
\newcommand\calF{{\mathcal F}}

\newcommand\calQ{{\mathcal Q}}
\newcommand\Q[1]{{\mathcal Q}_{#1}}
\newcommand\justQ[1]{\overline{\mathcal Q}_{#1}}
\newcommand\justq{\overline{q}}
\newcommand\Qfive{\Q5}
\newcommand\qfive{q_5}
\newcommand{\Qonebig}{\Q{1,\text{big}}} 
\newcommand{\qonebig}{q_{1,\text{big}}}
\newcommand{\Qmatch}{\Q{\text{match}}}
\newcommand{\qmatch}{q_{\text{match}}}

\newcommand{\Qonefive}{\Q{1,5}}
\newcommand{\Qtwofive}{\Q{2,5}}

\newcommand\X{{\mathcal X}}

\newcommand\rel{{r}}
\newcommand{\FFcond}{\textsc{FFcond}}

\newcommand{\offbin}{12}
\newcommand{\onlbin}{18-2\eps}
\newcommand{\topLB}{10+6\eps}
\newcommand{\bigLB}{9-\eps}
\newcommand{\bigLBf}{9-\eps}

\newcommand{\bigfree}{8-8\eps}

\newcommand{\smallslot}{6-6\eps} 
\newcommand{\qtwoUB}{2(\ffsmall)}
\newcommand{\offhalf}{6}
\newcommand{\domslot}{12+4\eps}
\newcommand{\ffsmall}{4-4\eps}
\newcommand{\smallUB}{3-3\eps}

\newcommand{\quarter}{\ffsmall}
\newcommand{\niceLB}{5+\eps} 
\newcommand{\nicetime}{\bindiff}
\newcommand{\bindiff}{6-2\eps}
\newcommand{\largeLB}{6+2\eps}
\newcommand{\fant}{13-3\eps}
\newcommand{\fantoverload}{1-3\eps}
\newcommand{\nicefill}{15+3\eps}
\newcommand{\nicenobig}{6+4\eps}
\newcommand{\niceminustwelve}{3+3\eps}
\newcommand{\niceff}{2+6\eps}
\newcommand{\niceoverffsmall}{11+7\eps}
\newcommand{\nicerzero}{2+8\eps}
\newcommand{\fantandendbig}{4-2\eps}
\newcommand{\ENDBIG}{3+\eps} 
 
\newcommand{\initd}{\frac{\bigLB}{\topLB}}
\newcommand{\topmbig}{1+7\eps}
\newcommand{\dsize}{\frac{\topmbig}{\topLB}}
\newcommand{\smallUBF}{\frac{10+6\eps}3}
\newcommand{\quarterUBF}{4+\eps}
\newcommand{\heavyUBF}{\frac{9-\eps}{2}}
\newcommand{\bigtwoUBF}{5+3\eps}
\newcommand{\halfUBF}{6+2\eps}
\newcommand{\largemUBF}{\frac{\onlbin-\beta}{2}}
\newcommand{\largepUBF}{9-\eps}

\newcommand{\quarterffull}{14-3\eps}
\newcommand{\wbarfull}{14+2\eps}
\newcommand{\wsqfull}{12+4\eps}

\newcommand{\qonefive}{q_{1,5}}
\newcommand{\qtwofive}{q_{2,5}}
\newcommand{\qconstant}{15}
\newcommand{\qcminusthree}{12}
\newcommand{\qcminusfour}{11}
\newcommand{\qcminusfive}{10}  
\newcommand{\ffmidconstant}{2}  
\newcommand{\countconstant}{14}
\newcommand{\xsize}{19}
\newcommand{\xsizewin}{2}

\newcommand{\Ipart}{I_{\textrm{partial}}}
\newcommand{\Ifut}{I_{\textrm{future}}}

\newcommand{\lc}{\ell}
\newcommand{\wb}{w_{\textrm{big}}}
\newcommand{\wq}{w_{\textrm{top}}}
\newcommand{\wn}{w_{\textrm{large}}}
\newcommand{\wf}{w_{\textrm{top}}}
\newcommand{\chalf}{\delta_{\textrm{half}}}
\newcommand{\clarge}{\delta_{\textrm{large}}}

\newcommand{\Dhalf}{\Delta_{\textrm{half}}}
\newcommand{\justDhalf}{{\overline{\Delta}}_{\textrm{half}}}
\newcommand{\Dlarge}{\Delta_{\textrm{large}}}
\newcommand{\Dnice}{\Delta_{\textrm{nice,1}}}
\newcommand{\Dtwonice}{\Delta_{\textrm{nice,2}}}

\newcommand{\Dthreenice}{{\cal N}}
\newcommand{\cnice}{n}

\newcommand{\conenice}{\delta_{\textrm{nice,1}}}
\newcommand{\ctwonice}{\delta_{\textrm{nice,2}}}
\newcommand{\Dsh}{\Delta_{\textrm{half}}^{+}}

\newcommand{\topthreat}{\textsc{TopThreat}}
\newcommand{\largethreat}{\textsc{LargeThreat}}
\newcommand{\topblock}{\textsc{TopBlock}}
\newcommand{\largeblock}{\textsc{LargeBlock}}
\newcommand{\bigblock}{\textsc{BigBlock}}
\newcommand{\tbs}{\abs{\textsc{TopBlock}}}

\newcommand{\bigthreat}{\textsc{BigThreat}}

\newcommand{\timing}{\textsc{Counted}}

\newcommand{\NEzero}{\mathcal{R}_0}
\newcommand{\Ezero}{\mathcal{E}_0}
\newcommand{\NEzeroc}{47}
\newcommand{\NEzerocplusone}{48}
\newcommand{\NEzerocplustwo}{49}
\newcommand{\mmin}{58380}
\newcommand{\mmintext}{60000}
\newcommand{\mmintwo}{3300}
\newcommand{\Tfirst}{T_{1}} 

\newcommand{\easytwo}{easy\xspace}

\newcommand{\stretchfactor}{R} 
\usepackage{xcolor}
\usepackage{hyperref}
\usepackage{cleveref}
\usepackage[figure]{hypcap}

\usepackage[normalem]{ulem}

\usepackage{todonotes}

\hideLIPIcs 

\title{Improved online load balancing with known makespan}
\author{Martin Böhm}{University of Wrocław, Poland}{boehm@cs.uni.wroc.pl}{https://orcid.org/0000-0003-4796-7422}{Research supported by National Science Centre in Poland under grant SONATA 2022/47/D/ST6/02864.}
\author{Matej Lieskovský} {Computer Science Institute of Charles
University, Faculty of Mathematics and Physics, Prague, Czechia}
{ml@iuuk.mff.cuni.cz}{https://orcid.org/0000-0002-0058-3133} {Partially supported by GAUK
project 234723, and GA \v{C}R project 24-10306S.}
\author{Sören Schmitt}{Department of Mathematics, University of Siegen, Germany}{soeren.schmitt@uni-siegen.de}{}{}
\author{Jiří Sgall}
{Computer Science Institute of Charles University, Faculty of Mathematics
and Physics, Prague, Czechia}
{sgall@iuuk.mff.cuni.cz}{https://orcid.org/0000-0003-3658-4848}
{Partially supported by GA \v{C}R project 24-10306S.}
\author{Rob {van Stee}}{Department of Mathematics, University of Siegen, Germany}{rob.vanstee@uni-siegen.de}{}{}

\authorrunning{M. Böhm and M. Lieskovsk\'y and S. Schmitt and J. Sgall and R. van Stee}

\Copyright{Martin Böhm and Matej Lieskovsk\'y and Sören Schmitt and Ji\v{r}\'i Sgall and Rob van Stee} 

\ccsdesc[500]{Theory of computation~Online algorithms}
\nolinenumbers
\keywords{Online algorithms, bin stretching, bin packing}

\EventEditors{John Q. Open and Joan R. Access}
\EventNoEds{2}
\EventLongTitle{42nd Conference on Very Important Topics (CVIT 2016)}
\EventShortTitle{CVIT 2016}
\EventAcronym{CVIT}
\EventYear{2016}
\EventDate{December 24--27, 2016}
\EventLocation{Little Whinging, United Kingdom}
\EventLogo{}
\SeriesVolume{42}
\ArticleNo{23}

\title{Improved online load balancing with known makespan}

\begin{document}
\begin{titlepage}
\maketitle 

\begin{abstract}
    We break the barrier of $3/2$ for the problem of online load balancing with known makespan, also known as bin stretching. 
    In this problem, $m$ identical machines and the optimal makespan are given. The load of a machine is the total size of all the jobs assigned to it and the makespan is the maximum load of all the machines.
    Jobs arrive online and the goal is to assign each job to a machine while staying within a small factor (the competitive ratio) of the optimal makespan.
    
    We present an algorithm that maintains a competitive ratio of $139/93<1.495$ for sufficiently large values of $m$, improving the previous bound of $3/2$.
    The value 3/2 represents a natural bound for this problem: as long as the online bins are of size at least $3/2$ of the offline bin, all items that fit at least two times in an offline bin have two nice properties. They fit three times in an online bin and a single such item can be packed together with an item of any size in an online bin. These properties are now both lost, which means that putting even one job on a wrong machine can leave some job unassigned at the end. It also makes it harder to determine good thresholds for the item types. This was one of the main technical issues in getting below $3/2$. 

The analysis consists of an intricate mixture of size and weight arguments.
\end{abstract}

\end{titlepage}

\setcounter{page}{1}

\section{Introduction}
\textsc{Online Load Balancing with Known Makespan} is an online problem defined as follows. At the start of the input, the number $m$ of machines is revealed, followed by a sequence of jobs with sizes in $[0,1]$, arriving one by one. Each job needs to be assigned to a machine, and the load of a machine is the total size of the jobs assigned to it.
The algorithm is guaranteed a priori that the entire sequence of jobs can be scheduled on $m$ machines so that the makespan (the load of the most-loaded machine) is at most $1$. The objective of the algorithm is to schedule the jobs on the machines as they arrive, minimizing the makespan $\stretchfactor$ of the online schedule, which is allowed to be larger than $1$. The value $\stretchfactor$ is also known under the name \textit{stretching factor}.

The problem was first introduced in 1998 by Azar and Regev~\cite{azarregev98binstretch, azarregev01binstretch} under the name \textsc{Online Bin Stretching}, and studied intensively since~\cite{bsvsv17onepointfive, boehmsimon22,gabay2015online,kellerer2013efficient,   lhomme2022online}. 
Among its given applications is container repacking~\cite{bsvsv17threebins} and reallocation during a server upgrade.
This scheduling problem shares its terminology and some algorithmic ideas with \textsc{Online Bin Packing}. 
The overarching goal of the research of \textsc{Online Bin Stretching} and other related problems over the last few decades is to learn how a small amount of additional knowledge ahead of time (such as knowledge of the makespan) impacts the best possible competitive ratio for the quintessential online problem \textsc{Online Load Balancing}~\cite{graham1966}. 

To that end, another closely related problem is \textsc{Online Load Balancing with Known Sum of Processing Times}, where we have a guarantee that the total volume of jobs is at most $m$, but the optimum can be larger than $1$. (e.g., if jobs larger than $1$ appear in the input sequence). For comparison in \textsc{Online Bin Stretching} we have a guarantee on the makespan which is stronger, while in the classical \textsc{Online Load Balancing} problem we have no guarantee. Having information on the total volume of jobs or the makespan could be viewed as particular kinds of \emph{advice} given to the online algorithm~\cite{10.1145/2993749.2993766,dohrau2015online,renault2015online}.

To answer the general question above quantitatively, the state of the art is the following. For \textsc{Online Load Balancing}, Fleischer and Wahl~\cite{fleischerwahl00} presented a deterministic algorithm with competitive ratio approximately $1.92$, and Rudin~\cite{rudin} showed that no deterministic
algorithm can be better than $1.88$-competitive. 
Kellerer et al.~\cite{kellererkotovgabay2015} showed that having a guarantee on the sum of processing times allows an approximately $1.585$-competitive algorithm as $m$ goes to infinity, matching the lower bound of Albers and Hellwig~\cite{albers2012semi}. Finally, for \textsc{Online Bin Stretching}, Böhm et al.~\cite{bsvsv17onepointfive} presented an algorithm with stretching factor $3/2$, and Azar and Regev~\cite{azarregev01binstretch} showed that no algorithm can have a stretching factor below $4/3$.

\paragraph*{Our contribution.} We propose an online algorithm for \textsc{Online Bin Stretching} that is able to surpass the $3/2$ threshold:

\begin{theorem} For $m\ge\mmintext$ and for $\eps=1/31$, there exists an online algorithm for \textsc{Online Bin Stretching} with stretching factor $3/2 - \eps/6 = 139/93 < 1.495$. 
\end{theorem}
For $\eps=1/62$ the algorithm works already for $m\geq \mmintwo$.
Our algorithm builds upon the main concepts of its immediate predecessors~\cite{gabay2015online, bsvsv17onepointfive},
by keeping a portion of the bins empty until a later phase of the input, and by tracking combinatorial properties of the items using a \textit{weight-based analysis}. Any feasible algorithm must follow this general structure. However, once the stretching factor is set below $3/2$, new types of items appear which require great care to pack efficiently. See Figure \ref{fig:algcomparison} and Section \ref{sec:ideas}. The level of complexity of our algorithm as well as its analysis significantly surpasses the previously best-known results.
For instance, it now becomes necessary to use new item types when we start to fill up previously used bins later in the algorithm, as most of the initial item types do not fit well in the remaining space.
Achieving a ratio below $3/2$ for \emph{all} values of $m$ seems to be much harder still, as we often have constantly many bins which are only half-full; only when $m$ is large is the number of such bins negligible.

\paragraph*{History and related work.} The first results on \textsc{Online Bin Stretching} have appeared even before the introductory paper of Azar and Regev in 1998; a year before, Kellerer et al.~\cite{kellererkotov97} already discovered the matching lower and upper bound of $4/3$ for the case $m=2$. Since the beginning, some research works focus on the stretching factor for general number of bins $m$, while others focus on the special cases for $m$ small and fixed. One interesting property of \textsc{Online Bin Stretching} with fixed $m$ is that both the best-known lower bounds~\cite{lhomme2022online}
and some algorithms~\cite{lieskovsky2022}
were designed using a computer-aided approach based on the \textit{Minimax} algorithm, as initially proposed by Gabay et al.~\cite{gabay2017lowerbounds}.

For any value $m\geq 2$ a general lower bound of $4/3$ comes from Azar and Regev~\cite{azarregev01binstretch}.
For $m=3$, the best-known algorithm is by Böhm et al.~\cite{bsvsv17threebins}. The remaining lower and upper bounds for the range $3 \le m \le 8$, listed in the table below, were designed by multiple variants of computer-aided search; the results are by Böhm and Simon~\cite{boehmsimon22}, Lhomme et al.~\cite{lhomme2022online} and Lieskovsk\'y \cite{lieskovsky2024latin}.
\begin{flushleft}
\vskip -1em
\begin{tabular}{@{}cccccccc@{}}
 $m$ & 3 & 4 & 5 & 6 & 7 & 8 & $\ge 9$ \\ \hline
 {Lower bound} & 1.365~\cite{boehmsimon22} & 1.357~\cite{boehmsimon22} & 1.357 \cite{boehmsimon22} & 1.363~\cite{lhomme2022online} & 1.363 \cite{lhomme2022online} & 1.363~\cite{lhomme2022online} & $1.\overline{3}$ \cite{azarregev01binstretch} \\ 
 {Upper bound} & 1.375 \cite{bsvsv17threebins} & 1.393 \cite{lieskovsky2024latin} & 1.410 \cite{lieskovsky2024latin} & 1.429 \cite{lieskovsky2024latin} & 1.455 \cite{lieskovsky2024latin} & 1.462 \cite{lieskovsky2024latin} & 1.5 \cite{bsvsv17onepointfive}
\end{tabular}
\end{flushleft}
\vskip -0.7em

For general $m$, Böhm et al.~\cite{bsvsv17onepointfive} presented the so far best algorithm in 2017 which achieves stretching factor $3/2$; this result was preceded by a long sequence of steady improvements on the algorithmic front, among others by Kellerer and Kotov~\cite{kellerer2013efficient} and Gabay et al.~\cite{gabay2015online}.

For some small fixed values of $m$, especially $m=2$, also specialized problems related to \textsc{Online Bin Stretching} have been investigated previously; 
for example, Epstein~\cite{epstein03bsrevisited} considered online bin stretching with two machines (bins) of uniformly related speed
and Akaria and Epstein~\cite{epsteinakaria22hierarchical}
considered online bin stretching on two bins with grade of service and migration.

\section{Structure of the algorithm}

From now on, as is common in the literature on \textsc{Online Bin Stretching} and because we are dealing with a packing problem, we refer to bins, levels of bins and items instead of machines, loads of machines and jobs, respectively.
Our initial setting is that the offline optimum bins have size $1$ and the bins usable by our algorithm have size $\stretchfactor = 3/2- \eps/6$. 

We assume that the number of bins $m$ is at least $\mmintext$. 
We scale the sizes of the bins such that an offline bin has size $\offbin$ and an online bin has size $\onlbin$. 
(We use $2\eps$ here so that half of the size of an online bin is a more convenient value.)
Our goal is to construct an algorithm which works for the largest possible value of $\eps$. We will eventually set $\eps=1/31$, but for an easier understanding of the relationships between the various values we will mostly use symbolic calculations until Section \ref{sec:finish}, where the exact value of $\eps$ becomes important. Scaling the offline bin size to 12 allows us to work with near-integer type thresholds, which is convenient. After scaling, the total size of the jobs on input is at most $12m$.

Our algorithm uses the algorithms Best Fit and First Fit as subroutines. These algorithms work as follows. Both algorithms open a new bin if the item does not fit into any existing bin.
Otherwise, Best Fit places an item in a bin where the item can still fit and that, after placement, leaves the least amount of remaining empty space in the bin. This means it uses the most-filled bin that can still accept the item.
First Fit always places the item in the first bin in which it will fit, using the order in which it opened the bins. 
In this paper, we will sometimes fix the order in which the bins are to be used in advance, namely if these bins already contain some items. This means that we are applying First Fit to variable-sized ``bins'' (the empty spaces in the actual bins). We give a proof for the performance of First Fit on variable-sized bins which may be of independent interest.

\subsection{Item types}
Our algorithm initially uses the following item types; once we start filling up the bins in the fill-up phase, it will be necessary to use different types because of the amount of space that will be left. 
We group some item types into supertypes. There are two intervals for small items, as these items are packed the same way. The three weighting functions are related to the three types of items that fit only once in an offline bin. Having three separate such types is a consequence of the existence of quarter items (see Figure \ref{fig:algcomparison} and Section \ref{sec:ideas}).

\medskip
 \begin{tabular}{c|ccc|ccc|cc}
		\textbf{Supertype} &   & - &   & \multicolumn{3}{c|}{middle} & \multicolumn{2}{c}{dominant}\\ \hline 
		\textbf{Type} & small & quarter & small & nice & half & large & big & top \\ \hline
		\textbf{Maximum size} & $\smallUB$ & $\quarter$ & $\niceLB$ & $\bindiff$ & $\largeLB$ & $\bigLBf$ & $\topLB$ & $\offbin$\\
		\textbf{Weight $\wq$} & 0 & 1 & 1 & 2 & 2 & 2 & 3 & 4\\
		\textbf{Weight $\wb$} & 0 & 0 & 0 & 2 & 2 & 2 & 4 & 4\\
		\textbf{Weight $\wn$} &0 & 0 & 0 & 0 & 2 &4 &4 & 4
	\end{tabular}

\medskip
We describe the ideas behind these type thresholds in detail in Section \ref{sec:ideas}. Here we describe some fundamental properties of the various (super-)types. See also \Cref{fig:algcomparison}.

Dominant items fit only once in an online bin. 
Nice items fit twice in an offline bin and can be placed in an online bin while still leaving room for another item of any size. (These items are indeed in principle nice to pack, but we still need to be very careful with them.) 
Half and large items fit twice in an online bin, but a large item cannot be packed together with a half item in an offline bin (due to the thresholds $\bindiff$ and $\largeLB$), whereas two half items \emph{may} fit together in an offline bin. 

In our algorithm, we will pack small items only to a level of $\smallslot$ at the beginning to leave room for one top item or two half items. 
As described above, using First Fit guarantees that more than $\ffsmall$ is packed in almost all bins that are packed like this. Of course we get the same guarantee for small items of size more than $\ffsmall$, and this is what motivates the upper bound $\ffsmall$ for quarter items. It is also the same guarantee that we will achieve on average for (a certain subset of) the bins with quarter items (some bins will contain two quarter items). A big item fits in an online bin with two quarter items, and this is the reason that the dominant items are divided into two types.

\begin{definition}
\label{def:threats}
	For a partial input $\Ipart$, 
	let the value $\topthreat$ (resp., $\bigthreat,$ $\largethreat$) be the maximum number of top items (resp., big items, large items) in $\Ifut$ so that $\Ipart\cup\Ifut$ can be packed in $m$ bins of size \offbin, and let $\topblock$ (resp., $\bigblock,$ $\largeblock$) be the set of bins that contain more than $\bindiff$ (resp., $\onlbin-(\topLB)=\bigfree, \bigLB$).
\end{definition}

For any packing of a partial input, we have $\topthreat\le\bigthreat\le\largethreat$ and $\largeblock\le\bigblock\le\topblock$.

\begin{lemma}
	\label{lem:threats}
	For any feasible input $I$ and weighting function $w\in\{\wq,\wb,\wn\}$, we have $w(I)\le 4m$. 
 For any $k\ge0$ and any partial input $\Ipart$:
\begin{itemize}
    \item if $\wq(\Ipart)\ge4k$, then $\topthreat \le m - k$,
    \item if $\wb(\Ipart)\ge4k$, then $\bigthreat \le m - k$,
    \item if $\wn(\Ipart)\ge4k$, then $\largethreat \le m-k$.
\end{itemize} 
\end{lemma}
\begin{proof}
    The bound $\wn(I)\le4m$ follows directly from the type thresholds (a large and a half item do not fit together in an offline bin). For the other two weighting functions, note that for an item $i$ of type $j$, the weight $\wq(i)=\lfloor \frac{5}{\offbin}s_j\rfloor$ and 
 $\wb(i)=2(\lfloor \frac{3}{\offbin}s_j\rfloor)$, where $s_j$ is the infimum size of an item of type $j$ (where the small items are split into two separate types for this calculation, one for each range of small items). Intuitively, $\wq$ counts the number of items larger than $\frac{\offbin}5$, that is, items that fit at most four times in an offline bin. 
Similarly, 
$\wb$ counts items larger than
$\frac{\offbin}3$, and multiplies the result by two. The bounds $\wq(I)\le4m$ and $\wb(I)\le4m$ follow.
\end{proof}

The following invariant is a necessary property of any feasible algorithm and we will maintain it and other invariants throughout the processing of the input.
\begin{invariant}
\label{inv:threats}
    We have $\topthreat\le m-\topblock$ and $\bigthreat\le m-\bigblock$.
\end{invariant}
We will \emph{not} be able to maintain $\largethreat\le m-\largeblock$ throughout the algorithm (not even in the starting phase). However, fortunately large items can be placed twice in an online bin. Since these items can have size up to $\bigLB$, a bin must be completely empty in order to guarantee that two large items may be packed in it.

\begin{figure}[t]
	\centering
	\includegraphics[width=\textwidth]{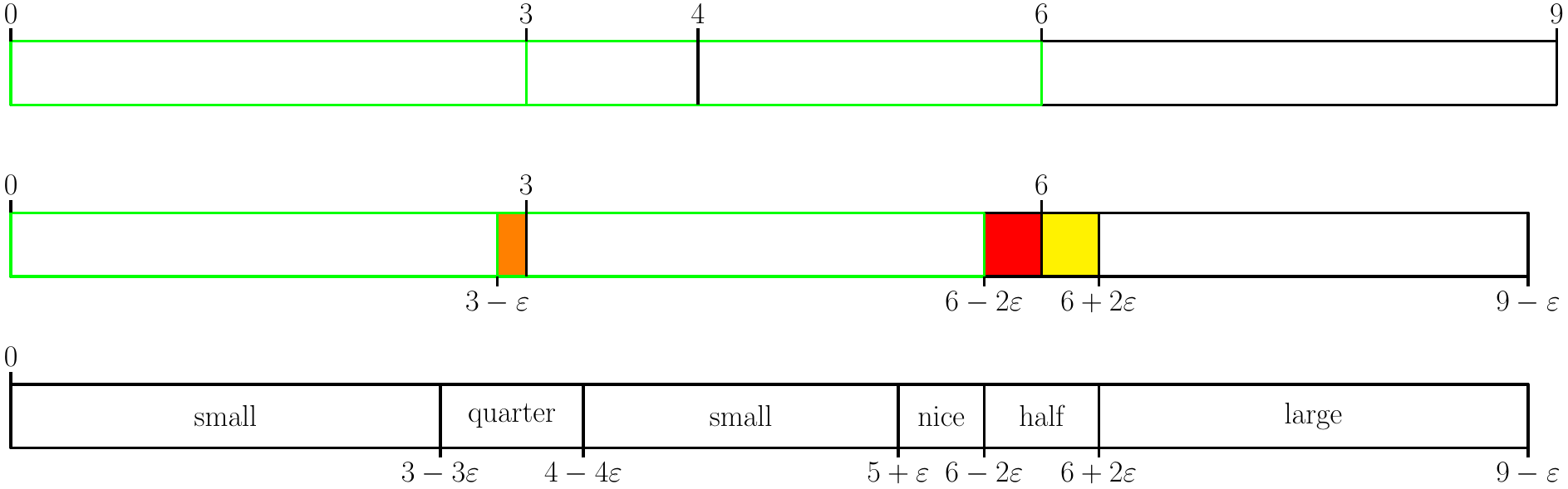}
	\caption{\label{fig:algcomparison}
	(Sketch, using $\eps=1/6$) A comparison of the important thresholds for an algorithm with competitive ratio $3/2$ (top) and an algorithm with competitive ratio strictly less than $3/2$ (middle). The thresholds our algorithm uses are displayed at the bottom. The offline bin size is scaled to be $12$, so all items in the input have size at most 12. 
 The green box indicates (half) the difference between the online and offline bin size. In the top figure, 6 is also a point where the amounts that can be packed in a bin change, both online and offline. \\
    We immediately see that in the middle figure 
    items exist which did not exist before (red); for a competitive ratio of $3/2$, the online algorithm can pack more items per bin for \emph{all} items smaller than 9. Moreover, items in the orange range can block some items of maximal size from being packed in the same bin, if we pack two such items in one bin. Finally, the fact that the red range exists means that items just larger than this (yellow) also need to be packed more carefully than before.}
    \label{fig:algcomparisoncap}
\end{figure}
\subsection{Phases and states}
\label{subsec:packing-strategies} 

\begin{figure}[t]
	\centering
	\includegraphics[width=\textwidth]{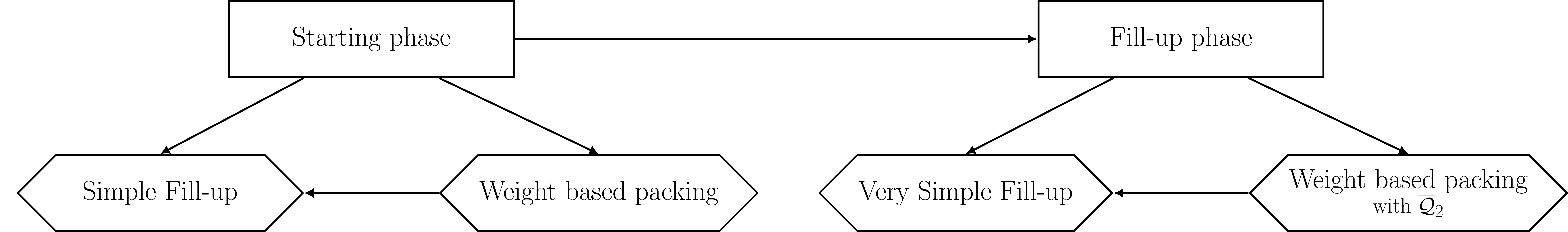}
	\caption{\label{fig:phases}
    An overview of the phases and states.}
 \end{figure}

In the starting phase, we use bins one by one, while staying below a level of $\bindiff$ unless there is a very good reason not to do so. If many relatively large items arrive, we may reach a state where it is sufficient to use First Fit for all remaining items (Simple Fill-Up) or where we know by weight that all items can be packed (Weight-based packing). Otherwise, we will eventually go to the Fill-up phase, where we start filling up the bins that previously received less than $\bindiff$ (or up to $8-8\eps$ in the case of bins with two quarter items). In this phase we will eventually also reach a state where we know that the remaining input can be packed, either by size or by a weight argument. 

\paragraph*{The starting phase}
From the lens of a single bin, our algorithm typically either packs items
until a bin is full -- which is typical for bins containing a single item type, such as the middle items -- or it packs them only up to a level of $\smallslot$, particularly for items of size at most $\smallslot$. 

However, as we have already seen, quarter items do not fit in this framework. On the one hand, we need to avoid packing many quarter items alone in bins (bad packing guarantee) while on the other hand, we also cannot pack too many quarter items in pairs in bins that do not yet contain anything else: that could block top items from being packed (Invariant \ref{inv:threats} would be violated).

Ideally, we would like to pack items as follows:
\begin{itemize}
	\item top items or pairs of half items with small items
	\item big items with quarter items 
	\item large items in pairs
    \item nice items three per bin
\end{itemize}

In this way, all bins would have a weight of at least 4 in $\wq$ and $\wb$ and they would also all be more than $\offbin$ full (except for bins that contain one big item and one quarter item and bins that contain a top item/two half items smaller than $\offhalf$ and not enough small items). There are several problems in using these methods, however:
\begin{itemize}
    \item For bins that are planned to contain items of two different types, or two items of one type, it is not known whether the second type or item will ever arrive.
    \item Packing large items and smaller middle items into separate bins can easily lead to instances that cannot be packed (if there are two bins with single middle items that fit together in an offline bin, and then many top items arrive).
\end{itemize}

We can work around the first problem by changing our packing methods after a certain number of bins have received items of only one type, in particular if many small or quarter items arrive. Basically, our algorithm will first aim to reach the ideal packing described above.
When sufficient volume has been packed, we go back and start filling up the already used bins. This is the fill-up phase of our algorithm. 

The second problem requires us to be very careful with nice items in particular, since some nice items fit with some large items in an offline bin. Packing nice items three per bin in dedicated bins will be fine. However, we cannot afford to do this already starting from the very first bin with nice items, as there could also be a bin with one half item and another bin with one large item at the same time, blocking too many bins for top items so that the algorithm fails. Fortunately, a bin with only a nice item can still receive an item of any other type, so we will pack one nice item alone before starting to pack them three per bin from the second bin onwards. We still need to be very careful if both half and nice items arrive.

\paragraph*{Good situations and the fill-up phase}
We may be fortunate and reach a situation where many bins are filled to (significantly) more than $\offbin$. In this case it will be sufficient to pack the remaining items by essentially using First Fit. This is one example of a \emph{good situation}. This is our term for a configuration which ensures that all remaining items can be packed, usually by using a very simple algorithm. This one is called the \textbf{First Fit case}. 

It may also happen that many relatively large items arrive early. In this case we may reach a state where we know that a small or quarter item will never need to be packed into an empty bin anymore, because they are packed in existing bins first and we would reach the First Fit case before using an empty bin. To ensure that the algorithm does succeed in all cases, even if all bins receive items, we will always use Best Fit as last resort for any item (after exhausting all other rules and all empty bins). We call this the \textbf{Rule of last resort}.\label{anyfit}

If many bins contain items but not enough of them contain a total size of more than 12 or sufficiently large items, it becomes important whether there exist bins that contain only small items or only single quarter items. If that is the case, we will go to the fill-up phase, in which we start filling up the nonempty bins using different item types. Otherwise, we will remain in the starting phase and we will eventually reach a good situation or the input will end.

\paragraph*{The $(9-\eps)$-guarantee} We need to determine when exactly it is safe to start filling up bins in which we have already packed some items, without failing for instance to the threat of top items. To be precise, once we start filling up bins, we need a guarantee that this \emph{remains} feasible no matter what the remaining input is. This will certainly require us to pack a sufficient total size in each bin that we fill up, as we always need to maintain Invariant \ref{inv:threats}.

Our cutoff for starting to fill up bins will be the point at which we know for certain that the future number of \emph{big} items is (and will remain!) strictly smaller than the number of bins in which big items can still be packed (so, $\bigthreat < m-\bigblock$).
There are in principle two ways by which we can know this: by considering weight and by considering volume. The problem with using a weight-based guarantee is that for instance small items can start arriving, which do not have weight. If we start filling up bins using small items, we can soon reach a point where the weight-based bound for $\bigthreat$ has not changed, but $\bigblock$ has increased and we fail when many big items arrive.

We therefore use a volume-based bound. We need to be careful also here. Suppose that already $2m/3$ bins contain small items, and each such bin has a level in the range $(\ffsmall,\smallslot]$. Now suppose that many big items start arriving one by one.  
These big items do \emph{not} bring us really closer to the point where we can safely start filling up the nonempty bins, because every time that we pack a big item $\bigthreat$ decreases by 1 and $m-\bigblock$ decreases by 1. Similarly, top items bring us only slowly closer to this point (since they are slightly larger than big items). 

We will start the fill-up phase once we know the so-called \textbf{$(9-\eps)$-guarantee} holds:
\begin{quote}
\fbox{Whenever new items of total size $\bigLB$ arrive, $\bigthreat$ decreases by at least 1.}
\end{quote}

Having the $(\bigLB)$-guarantee essentially ensures that packing $\bigLB$ per bin is sufficient to maintain Invariant \ref{inv:threats}, although the problem of $m$ large items arriving remains and needs to be dealt with separately. 
Maintaining this average is not at all straightforward, since we also have to make sure not to use too many \emph{empty} bins too early, in order to pack as many pairs of large items into them as possible.  

We present a very careful method of filling the nonempty bins which takes care to use the remaining space in those bins as efficiently as possible, using new item types which are tailored to the remaining space. This method consists of several stages.

\subsection{Bin types}
During the execution of the algorithm, each bin in the
instance will be assigned a specific type. 
Sets of bins of a certain (sub)type are denoted typically by script letters (possibly with an index). 
We define six main types of bins. We use the corresponding lower case letters to refer to numbers of bins of a type: for instance, $\lc=\abs{\wcomplete}$ and $\delta = \abs{\Delta}$. 
\begin{description}
\item[$\calE$] Empty bins.
\item[$\wcomplete$] Large-complete bins. This is a set of bins that reduce $\largethreat$; more generally, they reduce the number of items with weight that can still arrive. Specifically, the number of large items that can still arrive will always be at most $m-\lc$, and the total weight that can still arrive will be at most $4(m-\lc)$.
Since items without weight can still arrive however, and bins in $\wcomplete$ do not necessarily contain at least 12, large-complete bins do still accept items. A formal definition of these bins follows below.
\item[$\S$] Bins with only \emph{small} items. At most $\smallslot$ of small items is packed in each such bin.
\item[$\Delta$] At most four bins that contain two nice items or a single middle item and maybe some other items. See below for more details.
\item[$\calQ$] Bins that are not in $\Delta$ in which the first item (or the second, if the bin was previously in $\Delta$) is a quarter item and that are unmatched. (The algorithm sometimes matches some bins in $\calQ$; these bins are then moved to a special subset $\Qmatch$.)
\item[$\nice$] Bins in which the first two items are nice items and the third item is nice or half. 
\end{description}

Bins started by nice items are filled to triples of nice items in the ideal packing and kept separate to achieve this; these bins can become large-complete upon receiving a dominant item. 
With this large-scale picture in mind, the large-complete bins are defined as follows. These bins require a careful definition because nice items may exist.
\begin{definition}
\label{def:large-complete}
    A bin is called \emph{large-complete} if it satisfies all of the following conditions.
    \begin{itemize}
        \item it has $\wb\ge4$,
        \item it contains an item larger than $\offhalf$ or two items larger than $\bindiff$,
        \item the bin was never in $\calQ$.
    \end{itemize}
\end{definition}

It can be seen that each bin with $\wb\ge4$ has a big item or $\wq\ge4$. A large-complete bin does not necessarily contain a large item or a dominant item. The first condition ensures that these bins contain as much weight as any offline bin.
The second condition implies that $\largethreat\le m-\lc$ at all times. Note that this does not follow from the first condition alone, as a bin could contain two nice items, and a nice and a large item may fit together in an offline bin.

\medskip
The set $\Delta$ contains at most four exceptional bins used for careful handling of middle items. Each of these bins will be created explicitly in our algorithm if they are needed. There are the following four kinds of bins in $\Delta$. 
\begin{description} 
    \item[$\Dlarge$] one bin that contains a single large item, nothing else.
    \item[$\Dhalf$] one bin that contains a single half item and possibly small items of total size at most $\smallslot$ (notation $\Dhalf^S$) or a quarter item (notation $\Dhalf^Q$), nothing else. If the bin contains \emph{only} a half item we call it $\justDhalf$, else $\Dsh$.
    \item[$\Dnice$] at most two bins that contain one nice item and nothing else.
    \item[$\Dtwonice$] one bin that contains two nice items and nothing else.
\end{description}
\label{nice}
Our algorithm will use the so-called \textbf{nice rule} as long as possible: do not pack nice items into $\Dlarge\cup\Dhalf$ and do not pack half or large items into $\Dnice$. 
This rule ensures that nice items get packed into dedicated bins as much as possible (three per bin) so that we gain on these items (both by weight and by packed size per bin) compared to the optimal packing. This in turn ensures that nice items will hardly occur in inputs that are important for our analysis; see the weighting function $\wn$ and the proof of Good Situation \ref{gs:wb}.

\begin{definition}
    A bin is in $\calQ$ if it satisfies the following properties:
    \begin{itemize}
        \item The first item is a quarter item,
        \item The bin is not in $\Dhalf$,
        \item The bin has not been matched (see algorithm).
    \end{itemize}
    Additionally, a bin that was in $\justDhalf$ and then received a quarter item and finally another half or larger item is also in $\calQ$ as long as it has not been matched.
\end{definition}

We define the subset of $\calQ$ of bins in which the first two items are quarter items by $\Q2$, and $\Q1:=\calQ\setminus\Q2$. A bin in $\Q1$ may leave $\calQ$ (and $\Q1$) by receiving a half item (it enters $\Dhalf$); such a bin may later rejoin $\calQ$ by receiving a half or larger item. Bins in $\calQ$ may also leave $\calQ$ permanently by being matched (two bins in $\Q1$ to one bin in $\Q2$).

Instead of using the partition $\calQ=\Q1\cup\Q2$, we will also consider the useful partition of $\calQ$ in the table below. (Some of these subsets may be empty.)

\medskip 

\noindent
	\begin{tabularx}{\textwidth}{cX}
		Bin type & Conditions on contents\\ \hline
		\rule{0pt}{2.6ex}
		$\justQ1$ & A single quarter item, nothing else
  \\
  	$\justQ2$ & Two quarter items, nothing else\\
        $\Qonebig$ & First item is a quarter item, second item is big\\
		$\Qfive$ & First item is a quarter item,  $\wb\ge4$, bin is not in $\Qonebig$\\
   & Or: The first three items are (in this order) half, quarter, half or larger
	\end{tabularx}

\medskip 
We let $\justQ{}=\justQ1\cup\justQ2$.
Bins in $\Qfive$ can be in $\Q1$ (for instance, bins with a top item) or in $\Q2$; we keep track of their membership via the sets $\Qonefive:=\Qfive\cap\Q1$ and $\Qtwofive:=\Qfive\cap\Q2$. All bins in $\Qfive$ will have $\wq$-weight 5 (or more), explaining the name $\Qfive$. For comparison, bins in $\Q1$ have $\wq$-weight at least 1 and bins in $\Q2$ have $\wq$-weight at least 2.

Bins in $\Qonebig$ may get matched (pairwise) to bins in $\Qtwofive$; this is explained in the algorithm (Step 3). The set of matched bins is denoted by $\Qmatch$.

Since the first two items in each bin in $\Q2$ are quarter items, 
we have $\Q2\cap\Qonebig=\emptyset$.
We use the membership of bins in $\Q1$ and $\Q2$ to keep track of the distribution of quarter items in the non-large-complete bins. In our proofs, we will assign weight from the quarter items in $\Q1$ to bins in $\Q2$ on the one hand (so we need sufficiently many bins in $\Q1$) and assign volume from bins in $\Q2$ to bins in $\Q1$ on the other hand (so we need sufficiently many bins in $\Q2$). 
The separation from large-complete bins and the separation of $\Qonefive$ and $\Qtwofive$ will help us maintain an almost fixed ratio $q_1:q_2$.  Because of various half-full bins, we will need some additional bins in $\Q1$ (at most $\qconstant$ in the fill-up phase) before starting to create bins in $\Q2$.
A bin in $\justQ1$ that receives a half item leaves $\calQ$ and enters $\Dhalf$ (and $\Dhalf^Q$). If it later receives another half item or a large, it returns to $\calQ$, namely $\Qonefive$, or moves to $\wcomplete$.
Summarizing, we have the following disjoint unions.
\begin{align}
    \label{inv:q1}
    \Q1&=\justQ1\cup\Qonebig\cup\Qonefive\\
    \label{inv:q2}
    \Q2&=\justQ2\cup\Qtwofive.
\end{align}

We define the set of \emph{complete} bins $\calC$ as the set of bins that from the point of view of the algorithm (and the analysis) do not need to receive any specific items, as follows: 
\[
\calC:=\wcomplete\cup\Qfive\cup\Qmatch\cup\Dthreenice.
\]
Into these bins, any item may be packed. Finally, the unmatched nonempty bins that are not large-complete are called \emph{regular} (set $\relev$). We have
\begin{equation}
    \label{eq:relev}
    \relev = \S\cup\calQ\cup\nice\cup\Delta = 
    \S\cup \justQ1\cup\justQ2 \cup\Qonebig\cup\Qfive\cup \nice\cup\Delta.
\end{equation}
At all times, each bin is in exactly one of the sets $\relev,\Qmatch,\wcomplete,\calE$.

We will show eventually that in the starting phase, some bins remain empty or we can guarantee that all remaining items can be packed (possibly using different methods). However, the partitioning of the sets shown here remains valid even after we run out of empty bins apart from the fact that some bins may contain some items that do not belong there; for instance, there could be some small items packed into $\justQ2$. Also, after we run out of empty bins the nice rule may be violated.
We will maintain the following invariant.
\begin{invariant}
\label{inv:bigq}
    There is never a bin in $\justQ1\cup\justQ2$ at the same time as a bin with a big item as its only item with weight.
\end{invariant}

\begin{invariant}
    \label{inv:large}
    We have $\largethreat\le m-(c-n)$ as long as the nice rule is followed.
\end{invariant}
\begin{lemma}
\label{lem:invlarge}
    Invariant \ref{inv:large} holds for any packing of items.
\end{lemma}
\begin{proof}
    As long as the nice rule is followed,
    all complete bins except for the ones in $\nice$ contain two half items or an item larger than 6 and have $\wb\ge4$.
\end{proof}
From this bound it can be seen that the possible existence of bins in $\nice$ force us to keep bins empty for \emph{pairs} of large items, since we cannot ensure $\largethreat\le m-c$.

\subsection{Proof overview}
We begin with some initial observations regarding how many bins there can be of different types and how much they contain (Section \ref{subsec:mainanalysis}). We then focus on the set $\calQ$ in Section \ref{sec:Q}. We prove that up to an additive constant, $2q_2=q_1$ throughout the starting phase in Theorem \ref{thm:q1q2}. 
The (almost) fixed ratio $q_2:q_1$ is used to help show Invariant \ref{inv:threats} for top items in Theorem \ref{thm:q2weightbound}, and to show a packing guarantee for $\calQ$ in Theorem \ref{thm:qsizebound}.

There will be constantly many bins that do not satisfy our packing guarantees, these bins will be in a set $\X$. Packing guarantees for $\X$ and $\wcomplete$ are discussed in Section \ref{sec:LX}. 

In Section \ref{sec:firstphase} we show that in the starting phase, either some bins remain empty, $\justq_2>0$, or all items get packed. We begin by showing that Invariant \ref{inv:threats} is maintained as long as we do not use the rule of last resort (essentially, as long as some bins are empty) in Lemma \ref{lem:tte>0}. We then introduce some good situations in which we can guarantee that all remaining items can be packed (possibly using a different algorithm). Finally we show that Invariant \ref{inv:threats} is maintained in the entire starting phase or we reach a good situation. More generally, the algorithm does not fail in the starting phase (Theorem \ref{thm:firstphaseworks}). Along the way we show that packing $\bigLB$ additionally in each non-complete bin in the fill-up phase is enough to maintain Invariant \ref{inv:threats} in the fill-up phase as well (Theorem \ref{thm:finalphaseworks}).

We next explain how bins are filled up in the fill-up phase in Section \ref{sec:secondphase}. We first introduce some rules for filling up the bins and consider some sets of bins that are already very full in Section \ref{sec:prelim}. In order to fill up the remaining bins efficiently, we introduce new item type thresholds. We describe packing methods in Section \ref{sec:secondphasetypes}. Many items can be packed using First Fit, but quarter and large items require a careful packing in five stages (Section \ref{sec:hard}).

We analyze the fill-up phase in Section \ref{sec:finish}. We first consider some simple cases (essentially, new good situations) in Section \ref{sec:secondeasy}. In Section \ref{sec:firstthree} we show that the algorithm does not fail in the first three stages, and the rest is considered in Section \ref{sec:allitemspacked}.

\section{Algorithm in the starting phase}

Whenever the algorithm uses or attempts to use a set $A$ to pack an item in the following description, we use First Fit on the bins in $A$, unless otherwise stated. The notation $A\to B$ means that a bin in the set $A$ moves to $B$ by receiving an item of the current type.

\begin{description}
    \item[Step 1: Using and creating complete bins] Try the following in this order. 
\end{description}
\begin{itemize}
    \item for half and large items: $\Dtwonice\to\Dthreenice\subseteq\calC$, 
    \item for non-small items: $\Qonebig\to\Qonefive\subseteq\calC$
    \item use a complete bin (a bin in $\calC=\wcomplete\cup
    \Qfive\cup\Qmatch\cup\Dthreenice$).
    \item create a complete bin if this does not violate the nice rule (page \pageref{nice}). 
    
    First try the bins $\justQ2, \justQ1,\Dhalf^Q$ in this order.     
    Among other bins, use Best Fit to create a bin in $\wcomplete$, but do not pack a half item into $\Dlarge$ (yet).\footnote{If $\S\cup\justQ1\ne\emptyset$, we prefer packing half items there rather than in $\Dlarge$, because this improves certain packing guarantees later (see Theorem \ref{thm:sjustq1g0}.1). (If $\S\cup\justQ1\ne\emptyset$, $\Dhalf$ exists and a large item arrives, we will have that $\Dhalf=\Dhalf^S$ which already improves the guarantee.)}
\end{itemize}
\begin{description}
    \item[Step 2: Packing rules for each item type]
If an item is not packed yet, we apply the following rules depending on the item type. 
\begin{description}    
    \item[Small:] 
    First Fit on bins in $\S\cup\Dhalf^S$ while packing at most $\smallslot$ of small items in each bin,  $\justDhalf\to\Dhalf^S$, $\calE\to\S$. 
    \item[Quarter:]
    If $\abs{\Q1}+\chalf^Q\ge 2\abs{\Q2}+\qconstant$ then 
     $\justQ1\to\justQ2$, else $\justDhalf\to\Dhalf^Q$, $\calE \to \justQ1$.

    \item[Nice:] 
    $\Dtwonice\to\Dthreenice$, if $\conenice=2$ then $\Dnice\to\Dtwonice$, $\calE\to\Dnice$.
    \item[Half:] 
    Best Fit on bins in $\S\cup\justQ1\to\Dhalf,\Dlarge\to\wcomplete,\calE\to\Dhalf$. 
    \item[Large:] 
    $\calE\to\Dlarge$. 
    \item[Dominant:] 
    Always packed in Step 1, as will be explained in Section \ref{subsec:mainanalysis}.
\end{description}

\begin{description}
\item[Rule of last resort] 
If some item cannot be packed according to these rules, which can only happen after we run out of empty bins, we use Best Fit for this item, except that we still
follow the nice rule (see page \pageref{nice})
as long as possible. If 
the nice rule has already been violated,
we simply use Best Fit. For future items we still use the packing rules above first.
\end{description}

\item[Step 3: Matching rule]
This step minimizes the number of bins in $\Qonebig$. 

If $\abs{\Qonebig}\geq2$ and there is a bin in $\Qtwofive$, two bins in $\Qonebig$ are matched to a bin in $\Qtwofive$ and all three bins are moved from $\calQ$ to $\Qmatch$.

\item[Step 4: Swapping rule]
Each time that a new bin $\bar b$ in $\justQ{1}$ is created, if there exists a large-complete bin $b$ with a big item but no other items with weight (such a bin must contain also other items, or we would not have created $\bar b$), we virtually swap some items. 
That is, the bin $\bar b$ is treated as a bin in $\S$ from now on, and the bin $b$ supposedly contains a big item and a quarter item. This ensures that Invariant \ref{inv:bigq} is maintained. The quarter item is not considered to be the first item in $b$, so $b$ is not in $\calQ$. 

\end{description}

The swapping rule ensures that big items can safely be packed together with small items without violating Invariant \ref{inv:bigq} even if quarter items arrive later.
Whenever the swapping rule is applied on two bins $\bar b\in\justQ1$ and $b\in\wcomplete$, 
the total size packed into these bins is more than $\onlbin$ at this point (else $\bar b$ would not have been opened). If the bin $\bar b$ contains less than $\ffsmall$ we reassign volume such that the bin $\bar b$ ends up with exactly $\ffsmall$. We see that more than $\onlbin-(\ffsmall)=\wbarfull$ remains for the bin $b$. 
This just means that $\bar b$ possibly has slightly more space for additional items than the algorithm calculates with (because it views $\bar b$ as containing the small items that were in $b$). 

The large-complete bins used by the swapping rule differ from the other bins in $\wcomplete$ only in that they are not used to pack any future item. That is, we ignore such bins in Step 1 (this is not written explicitly in the algorithm; it seemed cleaner to explain this here). Our proofs will rely on the following observation.

\begin{observation}
    \label{obs:large-completeswap}
    At all times, the large-complete bins used by the swapping rule
    contain (or are assigned) at least $\wbarfull$. 
\end{observation}

\begin{figure}[ht]
	\label{fig:binoverview}
	\centering
	\includegraphics[width=\textwidth]{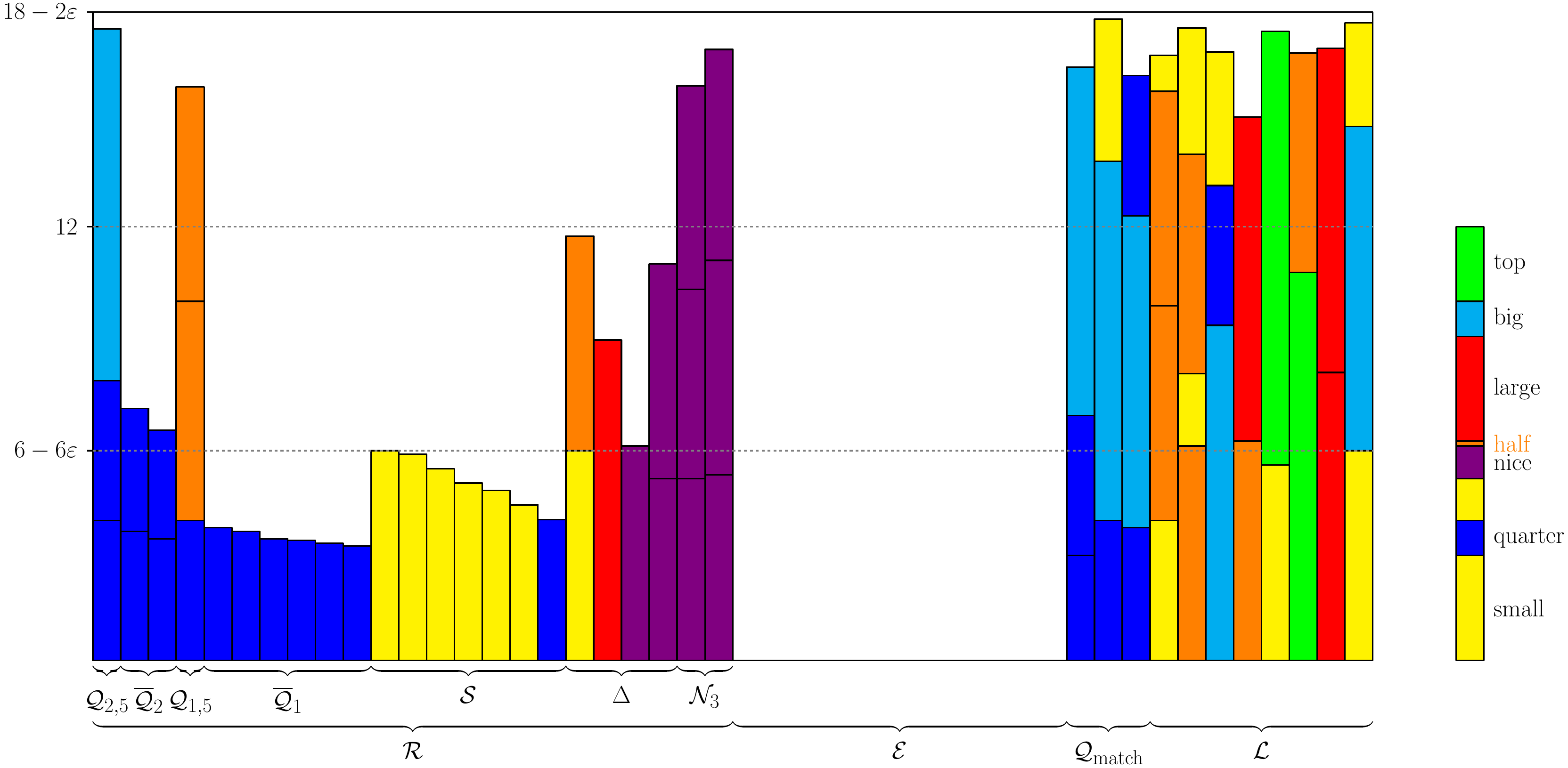}
	\caption{(Sketch) Overview of bins in the starting phase. The three bins in $\Qmatch$ were moved there by the matching rule. The second quarter item in the rightmost bin in $\Qmatch$ arrived there when the bins were already in $\Qmatch.$ The swapping rule was applied to the rightmost bin in $\S$ and the rightmost bin in $\wcomplete.$ The small items on top of the big item in the rightmost bin arrived before the swapping rule was applied. For visual clarity we have left out a number of bins in $\justQ1.$}
    \label{fig:binoverviewcap}
\end{figure}

\section{Properties of the bins in the starting phase}
\subsection{Bin type transitions and initial observations}
\label{subsec:mainanalysis}
We expand on our description
of the packing algorithm with some examples and observations. All dominant items will be packed in Step 1 -- their weight $\wb$ is equal to $4$ and they are larger than 6, so they can form a large-complete bin on their own. If they are packed into a bin in $\calQ$ they form a bin in $\Qfive\cup\Qonebig$. 

If big items arrive after some $\justQ2$ and $\justQ1$ bins
are formed, Step 1 will place them into said quarter item bins, with a preference for $\justQ2$. Symmetrically, if bins with big items exist, quarter items will be placed there (or in some large-complete bin) during Step 1 (if there is room) or by the swapping rule, if a large-complete bin with a big item and no other item with weight exists.
Otherwise, quarter items will form bins of type $\justQ1$ and $\justQ2$ (and at most one bin $\Dhalf^Q$). 

From the above discussion we learn that Invariant \ref{inv:bigq} is maintained. See below for a formal proof (Lemma \ref{lem:bigq}).
As long as $\bar q_2>0$, a bin in $\justQ1$ only receives an item of size more than $\offhalf$ if that item is a top item or if it is a half item. If a bin in $\justQ1$ receives a half item, the bin moves to $\Dhalf^Q$.
In $\Dhalf^Q$, only half items and larger items are packed. 

For middle items, we observe that any two such items together form a bin with $\wb\ge4$ (and $\wq\ge4$).
If there is a bin with only a single half item we pack a large item there (to create a large-complete bin) in Step 1. In general we want to pack half items in $\S\cup\calQ_1$ (ideal packing). When we do that, there can be one bin with one half item and some smaller items ($\Dhalf$) and also a bin with only one large item ($\Dlarge$) at the same time. The algorithm can still pack all remaining items (even if many top items arrive afterwards) because half and large items do not fit together in an offline bin. 

The arguments above also apply to the nice items, as they also belong to the middle supertype. The first nice item is treated specially, because a nice and a large item \emph{can} be packed together in an offline bin. However a nice item also has the (nice) property that any item, in particular any dominant item, fits with it in an online bin. This allows us to keep several $\Dnice$ bins open. Bins will only leave $\Dnice$ by receiving a dominant item (such bins are then large-complete) or a nice or half item (entering $\Dtwonice$).

A bin in $\wcomplete$ stays in $\wcomplete$; a bin in $\calQ$ stays in $\calQ$ except 
if a half item moves a bin in $\justQ1$ to $\Dhalf^Q$. Such a bin may return to $\calQ$ by receiving another half or large item.
In any step in which the rule of last resort is not applied, by the above discussion only the following bin type transitions are possible. 
\begin{itemize}
    \item A bin in $\justQ1$ can receive a quarter item ($\to\justQ2$), a half item ($\to\Dhalf^Q$) or a dominant item ($\to\Qonebig\cup\Qfive$).
    \item A bin in $\justQ2$ can receive a dominant item ($\to\Qtwofive$).
    \item A bin in $\S$ can move 
to $\wcomplete$ or to $\Dhalf^S$. 
\item A bin in $\justDhalf$ can receive a small item or a quarter item ($\to\Dhalf$), a half item ($\to\wcomplete$), or a large or dominant item ($\to\wcomplete$). 
\item A bin in $\Dnice$ can move to $\Dtwonice$ or $\wcomplete$. 
\end{itemize}
Apart from quarter items, bins in $\justQ{}$ only receive items of size at least $\bindiff$.
A \emph{regular} bin has $\wb\ge4$ if and only if it is in $\Qonebig\cup\Qfive\cup\Dtwonice\cup\Dthreenice$.
\begin{observation}
	\label{obs:twoquarter}
    There is never a bin in $\S\cup\justQ1$ at the same time as a large-complete bin containing {only} a dominant item or only two half items. 
\end{observation}

\begin{observation}   
    \label{obs:deltabounds}
    We have $\chalf\le1$ and $\clarge\le1$ at all times. We have
    $\conenice\leq2, \ctwonice\leq1$ and $\conenice+\ctwonice\leq2$. Hence $\delta\le4$ at all times.
    We also have $\bar\delta_{\text{half}}\cdot\clarge=0$ and $\bar\delta_{\text{half}}\cdot(s+\justq_1)=0$ at all times. Each bin in $\Dthreenice$ contains more than $\nicefill$. The total size of the items in each triple in $\Qmatch$ is more than $3\cdot(\bigLB)+4\cdot(\smallUB)=3\cdot(13-5\eps)$.
\end{observation}

\begin{lemma}
\label{lem:nofailq2}
    In the starting phase, the algorithm does not fail on an item that has size at most $\bindiff$, and as long as additionally $\justq_2>0$, the algorithm could only fail on a top item.
\end{lemma}
\begin{proof}
    If the algorithm failed on an item of size at most $\bindiff$, by the rule of last resort (page~\pageref{anyfit}) all bins would be more than $\offbin$ full, which cannot happen. 
    The algorithm will always use a bin in $\justQ2$ before failing completely, again by the rule of last resort. Any item that is not a top item fits in any $\justQ2$ bin.
\end{proof}

\begin{lemma}
\label{lem:newsq}
    If an item creates a new bin in $\S\cup\justQ1\cup\justQ2$, or a small or quarter item is packed into $\justDhalf$, all bins in $\calC$ that exist at that time contain (or are assigned) more than $(\onlbin)-(\niceLB)=\fant.$
\end{lemma}
\begin{proof}
  Bins in $\S\cup\justQ1\cup\justQ2$ are exclusively created by items of size at most $\niceLB$, that is, small and quarter items. By Step 1 the bins in $\calC$ are tried first, so all of them must contain more than $\fant.$
\end{proof}

\begin{lemma}
\label{lem:shlow}
    The number of bins in $\S\cup\Dhalf^S$ in which at most $\ffsmall$ of small items is packed is at most 2.
\end{lemma}
\begin{proof}
    All small items are of size either greater than $\ffsmall$ or at most $\smallUB$. All bins in $\S\cup\Dhalf^S$ are nonempty. These bins remain in the same order throughout the execution of the algorithm (but some bins may be moved to $\wcomplete$).
    Consider the first bin in $\S\cup\Dhalf^S$ that contains at most $\ffsmall$. (If there is none, we are done.) Since we use First Fit and pack up to $\smallslot$ of small items in each bin (and none of these bins contain more than $\onlbin-(\smallslot)$ of other items), any subsequent bin in $\S\cup\Dhalf^S$ contains only items of size more than $2-2\eps$. If any of these bins contains at most $\ffsmall$ of small items, they contain only a single such item with size in $(2-2\eps,\smallUB]$. Since we use First Fit and pack up to $\smallslot$, there can only be one such bin, meaning that there are at most two bins in $\S\cup\Dhalf^S$ that contain at most $\ffsmall$ of small items in total.
    This proves the lemma.
\end{proof}

The lower bound for the second range of small items was set to match the bound $\ffsmall$.

\subsection{Packing and weight guarantees for $\calQ$}
\label{sec:Q}

\begin{lemma}
\label{lem:bigq}
    Invariant \ref{inv:bigq} is maintained in the starting phase.
\end{lemma}
\begin{proof}
    The invariant holds at the start. As described above, arriving quarter items are packed with big items and 
    big items are packed into $\justQ1\cup\justQ2$ in Step 1 if possible. In particular, these options are considered before possibly creating any bin considered in Invariant \ref{inv:bigq}.
    Thus if a big item is packed outside of $\justQ1\cup\justQ2$, then $\justQ1\cup\justQ2=\emptyset$.
    A bin in $\justQ2$ can only be created from $\justQ1$. Due to the swapping rule, a bin in $\justQ1$ is not created as long as a bin with a big item as its only item with weight exists, so $\justQ1\cup\justQ2=\emptyset$ remains true until all such bins have been transformed using the swapping rule (or until the input ends).
\end{proof}

\begin{lemma}
\label{lem:qchanges} 
    The sum $q_1+\chalf^Q+q_2$ only decreases by application of the matching rule; this is also the only way to decrease $q_2$.
    The value $\justq_2$ can decrease at any time. 
   The value $q_1+\chalf^Q$ can only increase if $\justq_1+\chalf^Q$ increases on arrival of a quarter item. The number $q_2$ can only increase if $\justq_2$ increases; at such a time $\justq_1$ decreases (and $\chalf^Q$ remains constant). Apart from the matching rule, this is the only way to decrease $q_1+\chalf^Q$.
\end{lemma}
\begin{proof}
The only option for a bin to leave $\calQ$ besides the matching rule 
is 
if 
it 
is 
in 
$\justQ1$ and receives a half item (keeping $\justq_1+\chalf^Q$ constant). Bins can transition from $\Q1$ to $\Q2$, but any bin joining $\Q2$ will
stay there unless it gets matched. A bin can leave $\justQ2$  by receiving a dominant item (and this can happen at any time). This shows the first two statements.

The value $q_1$ only increases if a new quarter item is packed into an empty bin, increasing $\justq_1$, or if the bin from $\Dhalf^Q$ comes back to $\Q1$, which keeps $q_1+\chalf^Q$ constant.

The value $\chalf^Q$ can increase on arrival of a quarter item. If a half item causes the increase, $\justq_1+\chalf^Q$ remains constant as discussed above. As these are the only options for increasing these values, the third statement follows.

Since the first two items in each bin in $\Q2$ are quarter items,
the only way to create a new bin in $\Q2$ is by packing a quarter item in a bin in $\justQ1$, increasing $\justq_2$ and decreasing $\justq_1$. Regarding $q_1 + \chalf^Q$, we also have the possibility of a bin leaving $\Dhalf^Q$; this happens when this bin receives a half item or larger. It then enters $\Qfive$, or, more specifically, $\Qonefive$. This change thus sets $\chalf^Q$ to zero but increases $q_1$ by one, so $q_1 + \chalf^Q$
does not decrease.
\end{proof}

\begin{theorem}
    \label{thm:q1q2}
    In the starting phase, at all times we have 
    \begin{equation}
        \label{eq:q1inv}
        q_1+\chalf^Q \le 2q_2+\qconstant    
    \end{equation}
     and if $q_2>0$ we also have
    \begin{equation}
        \label{eq:weightinv}
        2q_2+\qcminusthree\le q_1+\chalf^Q.
    \end{equation}
\end{theorem}
\begin{proof}
    At the start, $q_1=q_2=0$. Applications of the matching rule do not invalidate the claimed bounds. 
    From \Cref{lem:qchanges} we know that $q_1 + \chalf^Q$
    increases only when $\justq_1 + \chalf^Q$ increases on arrival of a quarter item. As the packing rule for quarter items increases $\justq_1 + \chalf^Q$ only when $q_1+\chalf^Q < 2q_2+\qconstant$, the inequality (\ref{eq:q1inv}) is maintained as
    $q_1 + \chalf^Q$ increases. As $q_2$ is nondecreasing whenever the matching rule is not applied (again by \Cref{lem:qchanges}), the inequality (\ref{eq:q1inv}) holds throughout the starting phase.

    We turn our attention to inequality (\ref{eq:weightinv}), which is only valid when $q_2 > 0$. We focus on the moment just before
    $q_2$ increases, when the quarter item is about to be placed. \Cref{lem:qchanges} tells us that this happens only when $\justq_2$ is about to be increased. Combining inequality (\ref{eq:q1inv}) with the packing rule for quarter items, which says that at this time, we have $q_1+\chalf^Q \ge 2q_2+\qconstant$, we conclude that just before the adjustment, we have that $q_1 + \chalf^Q = 2q_2 + \qconstant$.

    If we denote $q'_2$ and $q'_1$ as the updated values after the quarter item is placed, we know that $q'_2=q_2+1$ and $q'_1=q_1-1$. We thus have
    \[2q'_2 + \qcminusthree = 2(q'_2-1) + \qcminusthree + 2 = 2q_2 + \qconstant -1 = q_1 + \chalf^Q -1
    = (q_1 - 1) + \chalf^Q = q'_1 + \chalf^Q. \]

    Thus, whenever $q_2$ increases, the inequality (\ref{eq:weightinv}) holds with equality. This is sufficient by the last statement of  \Cref{lem:qchanges}.
\end{proof}

\begin{lemma}
    \label{lem:q1b}
    If $\justq_2>0$, then $\qonebig=0$. At all times $\qonebig\le\qconstant$.
\end{lemma}
\begin{proof}
    Consider a time at which $\justq_2=0$ and we create a bin in $\justQ2$. At this time $\qonebig=0$ by Step 1 of our algorithm.
    After this, $\qonebig$ can only increase if $\justq_2=0$ again, because for big items the algorithm prefers $\justQ2$. Thus if $\qonebig$ increases we always have $\Q2=\Qtwofive$. By the matching rule (step 3 of the algorithm), we match two bins in $\Qonebig$ to a bin in $\Qtwofive$ whenever possible. Hence as long as $\qtwofive>0$, we have $\qonebig\le1$. If all bins in $\Qtwofive$ are matched and $\qonebig$ increases, then $\qonebig\le q_1-2\qtwofive=q_1-2q_2\le\qconstant$ by Theorem \ref{thm:q1q2}.
\end{proof}

The inequality (\ref{eq:q1inv}) will allow us to reassign packed size to the bins in $\justQ1$ later and reach an average packed size of (almost) $\ffsmall$.  
The first inequality is used to reassign weight in the next theorem.

\begin{theorem}\label{thm:q2weightbound}
    If $\justq_2>0$ and 
    no bin has entered $\topblock$ by the rule of last resort, 
    then $q_1\ge2q_2+\qcminusfive$ 
    and it is possible to reassign weight so that 
    \[
    \wq(\topblock\cup\Dnice)\ge4\abs{\topblock\cup \Dnice}
    \] 
    and the $\wq$-weight of all other bins is at least 0. In particular, we have $\topthreat\le m-\abs{\topblock\cup\Dnice}$ whenever $\justq_2>0$.
\end{theorem}
\begin{proof}
    The bins in $\Qonefive$
    contain $\wq$-weight 5. The large-complete bins used by the swapping rule have assigned $\wq$-weight 4, as they contain a big item and have a quarter item assigned to them which is not in $\Q{}$. 
    By the assumption and Invariant \ref{inv:bigq}, every bin that contains a big item also contains (or is assigned) another item with $\wq$-weight. We see that all bins in $\calC\cup\Dtwonice$ contain (or are assigned) $\wq$-weight 4 already. 
   
    We next show that if $\justq_2>0$, we can reassign weights such that $\wq(\relev\cup\wcomplete)=\wq(\S\cup \justQ1\cup\justQ2 \cup\calC\cup\Delta)\geq4(\justq_2+c+\delta)$ and $\wq(\S\cup\justQ1)\ge0$.
    The equality follows from (\ref{eq:relev}) and the fact that $\qonebig=0$ if $\justq_2>0$ (Lemma \ref{lem:q1b}). If all bins are nonempty and some items are placed into bins where they do not belong, the weight of those bins is not lower than the bounds given here. Further, in this case we have $\topblock\subseteq \justQ2\cup\calC\cup\Delta\setminus\Dnice$. 

    Bins in $\justQ2\cup\Delta$ contain $\wq$-weight at least 2.
    By Observation \ref{obs:deltabounds}, $\delta\le4$, so at most $2\justq_2+8$ weight is missing in $\justQ2\cup\Delta$.
    By (\ref{eq:weightinv}), we have $q_1+\chalf^Q\ge2q_2+\qcminusthree$ if $\justq_2>0.$ 
    We reassign 1 weight from at most $2\justq_2+8 \le 
    2q_2+\qcminusfour$ 
    bins in $\Q1$ to $\justQ2\cup(\Delta\setminus\Dtwonice)$
    to get the claimed result (each bin receives at most 2 weight). This maintains $\wq(\Qonefive)\ge 4\qonefive$ and $\wq(\justQ1)\ge0$. The bound on $\topthreat$ follows from Lemma \ref{lem:threats}.
\end{proof}

From this point, many of the following packing guarantees apply to all except constantly many bins. For that reason we introduce a set $\X$ of the following exceptional bins.

\begin{itemize}\label{def:X}
    \item $\X_s$: the bins in $\S$ filled less than $\ffsmall$ or bin in $\Dhalf^S$ with less than $\ffsmall$ of small items.
    \item $\X_\lc$: the bins in $\wcomplete$ that have level less than $\wsqfull$. 
    \item $\X_q$: a set that remains if we remove a subset containing $2q_2$ bins from $\Q1\cup\Dhalf^Q$, preferring bins in $\Qonebig$ to be removed. 
    We have $x_q\le\qconstant$ by (\ref{eq:q1inv}).
    The set $\X_q$ is initialized when we first create a bin in $\justQ1$ and updated when the numbers $q_1$ or $q_2$ change (by creating new bins in $\justQ1\cup\justQ2$). 
    \item $\Dnice\cup\Dtwonice$.
\end{itemize}
The sets that make up the set $\X$ are only considered for the analysis. 
For packing decisions, the algorithm considers the bins in $\X_q$ to still be in $\Q1$.

\begin{theorem}
\label{thm:qsizebound}
    At any time, we can reassign volume among bins in $\calQ\setminus\Qmatch$ such that we have the following lower bounds for the levels in the various subsets of  $\calQ.$ 
       \[\begin{array}{rccc}
           A & \Q1\setminus(\Qfive\cup\Qonebig)
           &  \Qonefive & \Qonebig
           \\ \hline\rule{0pt}{2.6ex}
        \text{Lower bound for }
        A\cap\X_q
        & \smallUB 
        & (\smallUB)+(\topLB) &(\smallUB)+(\bigLB) 
        \\
        \text{Lower bound for }
        A\setminus\X_q
        & \ffsmall 
        & (\ffsmall)+(\topLB) & 
        \\
        \multicolumn{1}{c}{}\\
        A &  \justQ2& \Qtwofive 
        \\ \hline\rule{0pt}{2.6ex}
    \text{Lower bound for }A        
    &\ffsmall &(\ffsmall)+(\bigLB) 
    \end{array}
    \]
\end{theorem}

\begin{proof}
    Consider the first column and the first summands in the other columns. These entries correspond to the quarter items. Each bin in $\Q1$ contains at least one quarter item and each bin in $\Q2$ contains at least two quarter items. By definition we have $\X_q\subseteq\Q1\cup\Dhalf^Q$. As each quarter item is of size at least $\smallUB$ the values for the first column and the first summands in the first line are explained. This leaves the four appearances of the number $\ffsmall$ in the table unexplained. We have $q_2$ bins in $\Q2$ and exactly $2q_2$ bins in $\Q1\setminus\X_q$ by the definition of $\X_q$ and invariant~(\ref{eq:q1inv}). Therefore $q_1-x_q=2q_2$.
    Each $\Q2$ bin contains quarter items of total size more than  $\smallslot$ and each bin in $\Q1\setminus\X_q$ contains a quarter item of size more than $(\ffsmall)-(1-\eps)$. Hence we can reassign $2-2\eps$ from one bin in $\Q2$ away such that two bins in $\Q1\setminus\X_q$ receive $1-\eps$ each, for an average of $\ffsmall.$ 

    This leaves only the second summands in the second column unexplained. It follow from the fact that the bins in $\Qonefive$ contain one top item or two half items and bins in $\Qtwofive$ contain one dominant item. The bins that contain one quarter item and one big item are always either in $\Qonebig\subseteq(\Q1\setminus\Q5)$ or in $\wcomplete$.
\end{proof}

\subsection{Packing guarantees for $\wcomplete$ and $\X$}
\label{sec:LX}

\begin{theorem}
\label{thm:sjustq1g0} 
Suppose $s+\justq_1>0$. Then the following statements hold.
\begin{enumerate}
    \item A bin in $\wcomplete$ with level at most $\wsqfull$ can only be created from a bin in $\X_s$. Such a bin is then in $\X_\lc$.
    \item There are at most two bins in $\X_s\cup\X_\lc.$
    \item All bins in $\wcomplete$ contain more than $13-5\eps$, are in $    \X_\lc$ or contain two large items.
    \item There are at most $\xsize$ bins in $\X$. At most $\xsizewin$ of them can be large-complete, namely the ones in $\X_\lc$. 
    Each bin in $\relev\setminus(\X\cup\Dthreenice)$ contains on average at least $\ffsmall$. 
\end{enumerate}
\end{theorem}
\begin{proof}
\begin{enumerate}
    \item 
    Suppose a bin $b$ enters $\wcomplete$ on receiving some item while $s+\justq_1>0$, and is not in $\X_s$, nor was it in $\calQ$. 
    An empty bin would have to receive a dominant item to enter $\wcomplete$, but this does not happen if $s+\justq_1>0$. Therefore $b$ already contained more than $\ffsmall$, either because it was in $\S\setminus\X_s$ or because it was in $(\Delta\setminus\justDhalf)\cup\Dnice$. We have that $b\neq \justDhalf$ by Observation \ref{obs:deltabounds}. If $b=\Dhalf^Q$, only a half or large item would be packed into it, for a total packed size of more than $2\cdot(\bindiff)+\smallUB=15-7\eps$.
    
    We now make a case distinction based on the last item packed into $b$.
    A dominant item makes the size packed into $b$ larger than $13-5\eps$. As long as $s+\justq_1>0$, a large item is not packed into $\Dnice$ because that would violate the nice rule.\footnote{This is where we use the restriction on the rule of last resort once no bin remains empty.} If it is packed into $\Dlarge\cup\Dhalf^S\setminus\X_s$, the total size is more than $\offbin+4\eps$.
    A half item only creates a large-complete bin if it is larger than $\offhalf$, if it is packed into a bin with an item larger than $\offhalf$ or if it is packed with another half item. If the other non-small item in the bin is a large item, the algorithm does not pack the half item in it because it prefers using a bin in $\S\cup\justQ1$ over $\Dlarge$, and with a dominant item we get the bound $15-3\eps$. Else, the bin contains more than $\ffsmall$ of small items if $s+\justq_1>0$, for a total size of more than $16-6\eps$.
    
    Finally, a nice item is not packed into $\Dlarge$ or $\Dhalf$ (which is the only way it could make a bin large-complete) because that would violate the nice rule while $s+\justq_1>0$, and no bin can become large-complete by receiving a small or quarter item because these items do not increase $w_b$ and also do not affect the other conditions in Definition \ref{def:large-complete}.
\item A new bin with small items is only created if a small item is packed into an empty bin or into $\justDhalf$. By the packing rules at that time all bins in $\wcomplete$ are filled to at least $\fant>\wsqfull$, so $\X_\lc=\emptyset$. By Lemma \ref{lem:shlow} we have $\abs{\X_s}\leq2.$ A quarter item only creates a bin in $\justQ1$ if all bins in $\wcomplete$ are filled to at least $\fant$ as well.
    Hence the statement cannot become wrong if $s+\justq_1$ increases (this includes that a bin is added to $\X_s$). It could only become wrong if a bin is added to $\X_\lc$ while $s+\justq_1>0$, so a bin in $\wcomplete$ with load less than $\wsqfull$ is created while $s+\justq_1>0$. By the first statement this can only happen from bins in $\X_s$ and hence for every bin created in $\X_\lc$ one bin from $\X_s$ is removed. By the above a bin can only be added to $\X_s$ while $X_\lc=\emptyset.$ 
\item Reviewing the case distinction in the proof of the first statement, we see that the only way for a bin to enter $\wcomplete$ while $s+\justq_1>0$ without coming from $\calQ\cup\X_s$ and without having level more than $13-5\eps$ is for that bin to contain two large items. Conversely, if $s+\justq_1=0$ and this changes by packing a new small or quarter item into an empty bin, all bins in $\wcomplete$ contain more than $\fant$.
\item There are at most two bins in $\X_s\cup\X_\lc$ by the second statement, and from Observation~\ref{obs:deltabounds} we have $\conenice+\ctwonice\le2$. We have $x_q\le\qconstant$ by (\ref{eq:q1inv}).
Therefore $x\le\qconstant+2+2=\xsize$. 

We have $\relev=\S\cup\Q{}\cup\nice\cup\Delta$ by (\ref{eq:relev}). The statement about the level in $\S\setminus \X$ holds by Lemma~\ref{lem:shlow} and in $\calQ\setminus \X$ by Theorem \ref{thm:qsizebound}. \qedhere 
\end{enumerate}
\end{proof}

\section{The starting phase succeeds}
\label{sec:firstphase}
In this section we show that the algorithm does not fail to pack an item during the starting phase. In some cases, it will even manage to pack all items already in this phase. We will show in particular that Invariant 1 is maintained during the starting phase.
For this to hold, the starting phase cannot continue too long.
The $(9-\eps)$-guarantee will start to hold roughly after packing $(3+\eps)m$ (as discussed before, at the beginning it is clearly not true).
To make this more precise, we will keep track of how much volume we have already packed. For easier bookkeeping during the second phase, we will ignore the following amounts\footnote{We consider the level of a bin after possible reassignment due to the swapping rule.}:
\begin{itemize}
	\item Anything packed above $\smallslot$ in bins in $\justQ2$
    \item Anything packed above $\smallUB$ in $\X_q$.
	\item Anything packed in $\X\setminus\X_q$
\end{itemize}
In each bin in $\Qonefive\cup\Qonebig\cup\Qmatch\setminus\X$ the entire quarter item gets counted, plus exactly $9-\eps$ from the rest (the remaining content is ignored).
We define $\timing$ as the total size of the items that are not ignored. 

\begin{observation}
    All ignored items or parts of items are in bins in $\X\cup \calQ$.
\end{observation}

\begin{lemma}
\label{lem:tte>0}
    As long as no bin enters $\topblock$ by the rule of last resort, Invariant \ref{inv:threats} holds.
\end{lemma}
\begin{proof}
We have $\topblock=\calC\cup\justQ2\cup\Qonebig\cup(\Delta\setminus\Dnice)$ and all other bins are in $\S\cup\justQ1\cup\Dnice\cup\calE$ by (\ref{eq:relev}).
We have $\calC\cup\Dtwonice\subseteq\bigblock\subseteq\calC\cup\Dlarge\cup\Dtwonice$.

    Suppose $\justq_2=0$.
    If there exists a bin with $\wb$-weight at least 6 (including the case $\cnice>0$), it can be seen that $\wb(\calC\cup\Qonebig\cup\Dtwonice)\geq 4(c+\qonebig+\ctwonice)+2$ and $\wb(\Delta\setminus\Dtwonice)=2(\delta-\ctwonice).$ So in total we have $\wb(\calC\cup\Qonebig\cup\Delta)\geq4(c+\qonebig+\ctwonice)+2+2(\delta-\ctwonice)\geq 4(c+\qonebig+\delta-\conenice)-2=4\abs{\topblock}-2$ (because $\chalf+\clarge\leq 2$ and $\Dnice$ is not in $\topblock$). Hence we have for the remaining input $\wb(I_{\text{rem}})\leq 4m-(4\abs{\topblock}-2)$ and it follows 
    $\topthreat\leq \bigthreat\le\lfloor m-\abs{\topblock}+\frac{1}{2}\rfloor=m-\tbs\le m-\abs{\bigblock}.$

    If there exists no bin with $\wb$-weight at least 6, but $\conenice>0$, we get the same bound for $\wb(\calC\cup\Qonebig\cup\Delta)$ because bins in $\Dnice$ have $\wb$-weight 2 (and are not in $\topblock$).

    Finally, the case that there exists no bin with $\wb$-weight at least 6 (so $\cnice=0$) and $\conenice=0$. This implies that $\ctwonice=0$ because no bin has entered $\topblock$ by the rule of last resort. We consider the $\wn$-weight. We have $\wn(\calC\cup\Qonebig)\geq 4(c+\qonebig)$ and $\wn(\Delta)\geq 4\delta-2$ (as $\wn(\Dlarge)=4$ and $\conenice=\ctwonice=0$). We get the same bounds as before since the $\wn$-weight of big and top items is 4.

    For the case $\justq_2>0$ the bound on $\topthreat$ holds by Theorem \ref{thm:q2weightbound}. For $\bigthreat$ we can still use the bounds above even if $\justq_2>0$, because bins in $\justQ2$ are not in $\bigblock$ but also not in $\calC$. For example, we get $\wb(\calC\cup\Delta)\ge4\abs{\bigblock}-2$ in the first case (since $\qonebig=0$), so $\wb(I_{\text{rem}})\leq 4m-(4\abs{\bigblock}-2)$ and $\bigthreat\le m-\abs{\bigblock}$. 
    
    The only exception to this would be if a bin in $\justQ2$ entered $\bigblock$ by the rule of last resort. This only happens if $e=0$. By Theorem \ref{thm:q2weightbound} already $4m$ $\wq$-weight has been packed at this point. In this case $\bigthreat=\topthreat=0.$
\end{proof}

\subsection{Good situations}
\label{sec:FFcase}

We will next describe several cases in which we can guarantee to pack all remaining items in the input: the so-called \emph{good situations}. A good situation is a partial instance for which we can prove that a $18-2\eps$ capacity suffices until the end of the instance, and a proof of a good situation is then a presentation and analysis of the algorithm that finishes the partial instance.

A simple but important example is the following. 

\begin{goodsit}
\label{gs:weight}
    For the input $I$ packed so far, we have $\wq(I)=4m$ or $\wb(I)=4m$.
\end{goodsit}
\begin{proof}
    By Lemma \ref{lem:threats}, all remaining items have weight zero. By Lemma \ref{lem:nofailq2}, the algorithm does not fail on an item with zero weight, as all such items are smaller than $\bindiff$. 
\end{proof}

The next case can be reached from the starting phase as well as the second phase.
\begin{definition}
    For $\alpha>0$ we define  
\[
    \FFcond(\alpha)=\frac{\ENDBIG}{\ENDBIG+\alpha}m
    -\frac{\niceminustwelve-\alpha}{3+\eps}\cnice
    + \frac{2\cdot\offbin}{\alpha} +\frac{\ffmidconstant\cdot(4+8\eps)}{\alpha}+2.
\]
\end{definition}

\begin{theorem}[First Fit case]
\label{thm:ffcase}
    Let $\alpha>0$. Sort the bins in order of decreasing level. Define $\calF$ as a maximal prefix of the bins in this sorted order so that the bins in $\calF$ have average level at least $\offbin+\alpha$ and the last bin in $\calF$ has level at least $\offbin$.
    
    Suppose $\abs{\calF}\ge\FFcond(\alpha)$, 
    there are at most two bins not in $\calF$ that have level in $(\bigfree,\bigLB)$, and $\topthreat\le m-\tbs$. 
    Then all remaining items can be packed.
\end{theorem}

\begin{proof}
The idea of the proof is that if we use First Fit on the at most $m-\FFcond(\alpha)$ bins that are not in $\calF$, they will be filled to at least $\bigLB$ by Corollary \ref{cor:FF11}, for an average load of more than $\offbin$ if we run out of bins (showing that we cannot run out of bins). However, Corollary \ref{cor:FF11} can only be applied if all of those bins receive at least one item in the packing of First Fit. This may not hold if some of those bins already contain more than $\bindiff$. For this reason, the problematic bins for our proof are bins with level in the range $(\bindiff,\bigLB)$.

The bins that contain at most $\bigfree$ (including all the empty bins) are called useful and all other bins that are not in $\calF$ useless. Any item that is not a top item fits into each bin that is useful.
By assumption all but at most four bins are in $\calF$, useful, or useless but with level in $(\bigLB,\offbin+\alpha)$.
Let the number of useful bins be $m'$. 
We further partition the useful bins in two subsets. Set $A$ are the bins 
that have level in $(\bindiff,\bigfree]$, and set $B$ are the bins with level at most $\bindiff$ (including the empty bins). We have $\topblock=m-\abs{B}$.

We use the following algorithm to pack the items.
We sort the useful bins in order of decreasing level. This means all bins in $A$ come first, followed by the set $B$.
We run First Fit on the useful bins for all items together, ignoring the other bins. As long as bins in $A$ are still available, only top items will be packed into $B$.

For the analysis of this algorithm, we distinguish between two cases. 
If all bins in $B$ are used up first (and First Fit eventually fails), each bin in $A$ will receive at least one item (as all possible top items have already arrived since $\topthreat\le m-\topblock=\abs{B}$), and the total size packed into $A$ will be at least $((\onlbin)+(\bindiff))(\abs{A}-1)/2=
(12-2\eps)(\abs{A}-1)$ by  Corollary \ref{cor:FF11}, whereas bins in $B$ have received top items so each bin got more than $\topLB$.

If we run out of bins in $A$ first (this happens for instance if many big items arrive), then each bin in $A$ has received at least one item and we get the same bound as above for $A$. The bins in $B$ contain at least $\bigLB$ in the end (apart from at most one bin). We conclude that if First Fit fails to pack some item, we pack at least $(\bigLB) (m'-2)$ into these $m'$ bins. 

For the remaining calculations, we assign a packed size of $\bigLB$ to the first $m'-2$ bins and nothing to the last two bins.

We will show that using  
${2\cdot\offbin}/{\alpha}+\ffmidconstant\cdot(4+8\eps)/\alpha$ bins in $\calF$ to compensate for these two ``empty'' bins and the possibly two bins with level in $(\bigfree,\bigLB]$\footnote{The ``empty'' bins need to be compensated for at most $\offbin$ and the bins with level in $(\bigfree,\bigLB]$ need to be compensated for at most $\offbin-(\bigfree)=4+8\eps.$}, all $m$ bins together have an average level of at least $\offbin$, which shows that all items end up packed, contradicting that First Fit fails. We next explain this in more detail.

Let $C$ be the set of all but the last $\left\lceil \frac{2\cdot\offbin}{\alpha} +\frac{\ffmidconstant\cdot(4+8\eps)}{\alpha} \right\rceil$ bins in $\calF$. 

Suppose $\cnice=0$. By assumption and by the definition of $\FFcond(\alpha)$, $|C|\ge \frac{\ENDBIG}{\ENDBIG+\alpha}m+1$, and all these bins contain at least $\offbin+\alpha$ on average.
The number of useful and useless bins with (assigned) packed size more than $\bigLB$ (and the remaining bins in $\calF$, which have level at least $\offbin$) is therefore upper bounded by $\frac{\alpha}{\ENDBIG+\alpha}m-1 < \frac{\alpha}{\ENDBIG+\alpha}m$.

We conclude that $\frac{|C|}{m-|C|}>\frac{\ENDBIG}{\alpha}$, implying that $\alpha|C|>(\ENDBIG)(m-|C|)$. This means that the total surplus (size packed above $\offbin$) on the first $\abs{C}$ bins is more than the total deficit of the remaining bins at the time at which First Fit fails.
But then the total size of all items is more than $\offbin m$, a contradiction. 

We now consider the case $\cnice>0$. In this case there are $\cnice$ bins with an additional level of at least $\nicefill-(12+\alpha)=\niceminustwelve-\alpha$ above what was already considered.
This surplus compensates for an additional 
$\frac{\niceminustwelve-\alpha}{\ENDBIG}\cnice$
bins that may be packed to only $\bigLB$ by First Fit. Thus the bound $\FFcond(\alpha)$ decreases by $\frac{\niceminustwelve-\alpha}{\ENDBIG}\cnice$ 
compared to the case $\cnice=0$ as claimed.
\end{proof}

We define 
\begin{align*}
    \FFcond&:=\FFcond\left(\fantoverload\right)\\
    &=\frac{\ENDBIG}{(\ENDBIG)+(\fantoverload)}m
    -\frac{\niceff}{\ENDBIG}\cnice
    +O(1)\\
    &=\frac{\ENDBIG}{\fantandendbig}m-\frac{\niceff}{\ENDBIG}\cnice
    +O(1).
\end{align*}

The next good situation is a direct application of Theorem \ref{thm:ffcase}.
\begin{goodsit}[Simple Fill-Up]
    \label{gs:ffcase}
    Sort the bins in order of decreasing level. Define $\calF$ as a maximal prefix of the bins in this sorted order so that the bins in $\calF$ have average level at least $\offbin+\fantoverload$ and the last bin in $\calF$ has level at least $\offbin$.
    
    We have $\abs{\calF}\ge\FFcond$, at most four bins have level in $(8-8\eps,9-\eps]$ and $\topthreat\le m-\tbs$.
\end{goodsit}

The threshold $\niceLB$ between small and nice items is related to $\FFcond$ (since $\onlbin-(\niceLB)=\offbin+\fantoverload$) and could have been set (slightly) differently; the present value is convenient, e.g.~for the proof of Lemma \ref{lem:beta}.

\subsection{Ways to end the starting phase}
\begin{lemma}\label{lem:ttsq=0FFcond}
    If $s+\justq_1=0$ and $c\ge\FFcond$ before any bin enters $\topblock$ by the rule of last resort, we maintain $s+\justq_1=0$ and Invariant \ref{inv:threats} as long as we use the packing methods of the starting phase, including the rule of last resort.
\end{lemma}
\begin{proof}
    As long as $e>0$ the statement about Invariant \ref{inv:threats} holds by Lemma \ref{lem:tte>0}. Suppose we reach $e=0$, until this time the rule of last resort was not applied. If at this time $\justq_2>0$ then by Theorem \ref{thm:q2weightbound} already $4m$ $\wq$-weight has been packed and hence $\bigthreat=\topthreat=0.$ 
    
    Otherwise, it follows from $c\ge\FFcond$ and Good Situation \ref{gs:ffcase} that $s+\justq_1=0$ remains true because we pack small and quarter items into $\calC$ first. Even after packing items into $\Dnice$ using the rule of last resort, $s+\justq_1=0$ remains true for the same reason. In this case we also never pack any small or quarter items into $\Delta$.

    Now suppose $e=0$ and $\justq_2=0$ (which remains true because $c\ge\FFcond$) when $e=0$ was reached.
    If $s+\justq_1=0$, the only bins that are not in $\topblock$ are $\Dnice\cup\calE$. Hence the condition of Lemma \ref{lem:tte>0} is satisfied as long as no item is packed into $\Dnice$ by the rule of last resort, so Invariant \ref{inv:threats} holds. 

    We now consider an item $y$ that is packed in some existing bin in $\Dnice$ by the rule of last resort. At this point all bins are in $\calC\cup\Delta$ as there are no empty bins left.
    
    If $y$ is a dominant item we create a bin with $\wb$-weight at least 6. Putting this bin in $\calC$, we have $\wb(\calC\cup\Qonebig\cup\Delta)\ge 4(c+\qonebig+\ctwonice)+2+2\chalf+2\clarge\ge4\tbs-2=4m-2$ as in the proof of Lemma~\ref{lem:tte>0}.
    It follows that $\bigthreat=\topthreat=0.$
    
    If $y$ is a half or large item then at least $\Dlarge$ received additional 2 $\wb$-weight (or $\Dlarge$ did not exist to begin with) prior by the rule of last resort since we use Best Fit ($\Dlarge$ contains more than bins in $\Dnice$) and because we do not pack any small or quarter items into $\Delta$. So $\Dlarge$ entered $\calC$ or did not exist and we have 
    $\wb(\calC\cup\Qonebig\cup\Delta)\ge 4(c+\qonebig+\ctwonice)+2\chalf\ge4\tbs-2$ since $\chalf\le1$. We are done as above.

    Finally, if $y$ is nice (this can only happen by the rule of last resort if $\conenice=1$) then both $\Dlarge$ and $\Dhalf$ have previously received a weight of 2 (or did not exist) and these bins as well as $\Dnice$ can be put into $\calC$, so $\wb(\calC\cup\Qonebig)\ge4\abs{\topblock}$.  
    In all cases $\bigthreat=\topthreat=0$
    after $y$ is packed. 
\end{proof}

We define
\begin{align*}
\Tfirst(k) &:= (\ENDBIG)m +(\bigLB)(c+\qonebig+k),\\
\Ezero &:= \frac{1-5\eps}{\ffsmall}(m-\lc-\qmatch)+\frac{\nicerzero}{\ffsmall}\cnice + \frac{1-5\eps}{\ffsmall}\qmatch+ \frac{4\eps}{\ffsmall}\ell-\NEzerocplusone,\\
\NEzero &:= \left\lfloor(m-\lc-\qmatch)-\Ezero\right)\rfloor.
\end{align*}

\begin{theorem}[Foundation of the second phase]
	\label{thm:finalphaseworks}
	Let $k\ge0$. 
	If $\timing\ge \Tfirst(k)$, then at most a total size of $(\bigLBf)(m-c-\qonebig-k)$ 
    of items can still arrive in the future. If additionally $c+\qonebig+k>0$, then $\bigthreat<m-c-\qonebig-k$ and $\topthreat<\initd(m-c-\qonebig-k).$
\end{theorem}

\begin{proof}
	Suppose
	$ 
	\timing\ge \Tfirst(k).
	$ 
	The total size of the input is at most $\offbin m$. This leaves at most $(\bigLB)(m-c-\qonebig-k) 
 $ for any future and/or ignored items. 
 The bounds on $\bigthreat$ and $\topthreat$ follow.
\end{proof}

By applying Theorem \ref{thm:finalphaseworks} for successive values of $k$, we see that what we have called the $(\bigLB)$-guarantee continues to hold provided we pack enough into each bin: if we pack $\bigLB$ into one bin and we increase $k$ by 1, then $\bigthreat$ decreases by 1.

\paragraph*{Transitioning to the fill-up phase}
We will start the second phase if 
\[
\timing\ge \Tfirst(\countconstant)
\]
and we are not already in a good situation.
The number \countconstant\xspace will be explained later.
It is chosen to guarantee $\bigthreat< m-\abs{\bigblock}$ during the entire algorithm.

\medskip 
We define 
$\tilde m:=m-\lc-\qmatch$. From this point on we fix $\eps$ to be $1/31$ (but we will continue to use symbolic calculations as well).

\begin{theorem}\label{thm:emptywhensq>0}
    If $\timing < \Tfirst(\countconstant)$ and $s+\justq_1>0$, 
    then for $\eps=1/31$ we get $r\le\NEzero-1$, $e\ge\Ezero+1$ and $e=\Omega(m)$. 
\end{theorem}

\begin{proof}%
    We have that all nonempty bins are in $\relev\cup\Qmatch\cup\wcomplete$ and $\relev=\S\cup\calQ\cup\nice\cup(\X\cap\relev)$. 
	 In the bins in $\relev\setminus(\X\cup\Dthreenice)$ at least $\ffsmall$ is packed on average by Theorem \ref{thm:sjustq1g0}. The bins in $\Qfive\cup\Qmatch$ contain at least $(\ffsmall)+(\bigLB)$ after reassignment (Theorem \ref{thm:qsizebound} and  \Cref{obs:deltabounds}). 
     By Observation $\ref{obs:deltabounds}$, bins in $\Dthreenice$ contain at least $\nicefill=(\ffsmall)+(\niceoverffsmall).$ By Theorem \ref{thm:sjustq1g0} at most $\xsize$ bins are in $\X\cap\relev$. Of these bins, at most 4 are ignored in $\timing$, and at most $\qconstant$ are in $\Q1$.

    By Theorem \ref{thm:sjustq1g0} all bins in $\wcomplete\setminus\X$ have level at least $\wsqfull$ and $\abs{X_\lc}\leq\xsizewin$. 
    The worst case for our bound is if $\qconstant$ bins are in $\justQ1$, and $\qonebig=0$.
    We conclude that     
    \begin{align*}
		& (\ffsmall)(\rel-\xsize+\qmatch)+(3-3\eps)\qconstant
  +(\wsqfull)(\lc-\xsizewin)\\ &+ (\bigLB)(\qfive+\qmatch) +(\niceoverffsmall)\cnice 
  \\
		\le&\timing 
  <  
  (\ENDBIG) (m-\lc-\qmatch)+(\bigLB) (\countconstant+\qfive+\cnice) + \offbin (\lc+\qmatch).
	\end{align*} 
    It follows 
    \begin{align*}
    &(\ffsmall)\rel+ (3-3\eps)\qconstant+ 4\eps \lc+(\nicerzero)\cnice+(1-5\eps)\qmatch
    \\<
    &(\ENDBIG)(m-\lc-\qmatch) +(\ffsmall)\xsize
     +(\bigLB)\countconstant + (\wsqfull)\xsizewin,
    \end{align*} 
    so 
 \[r<\frac{3+\eps}{\ffsmall}(m-\lc-\qmatch)
-\frac{\nicerzero}{\ffsmall}\cnice- \frac{1-5\eps}{\ffsmall}\qmatch-\frac{4\eps}{\ffsmall} \lc + \NEzeroc. 
\]
As $r$ is integer it follows $$r\le \left\lfloor\frac{3+\eps}{\ffsmall}(m-\lc-\qmatch)
-\frac{\nicerzero}{\ffsmall}\cnice- \frac{1-5\eps}{\ffsmall}\qmatch-\frac{4\eps}{\ffsmall} \lc + \NEzeroc\right\rfloor=\NEzero-1.$$
We have $e+r=m-\lc-\qmatch$, so 
\begin{align*}
    e&\ge \frac{1-5\eps}{\ffsmall}\tilde m+\frac{\nicerzero}{\ffsmall}\cnice + \frac{1-5\eps}{\ffsmall}\qmatch+ \frac{4\eps}{\ffsmall}\ell-\NEzeroc=\Ezero+1\\&\ge\frac{1-5\eps}{\ffsmall}\tilde m+\frac12\cnice+\frac{1-5\eps}{\ffsmall}\qmatch+\eps \lc-\NEzeroc.
\end{align*}     If $(\lc+\qmatch)\ge m/2$, then $e\ge \eps m/2-O(1)$, else $e\ge \frac{1-5\eps}{\ffsmall} m/2-O(1)$. Here we use that $\frac{1-5\eps}{\ffsmall}>\frac{1}{5}>\eps$ for $\eps=1/31.$
\end{proof}

\begin{goodsit}[Weight-based packing]
\label{gs:wb}
    We have $c\ge\FFcond, e=\Omega(m)$ and $s+\justq_1+\justq_2=0$. 
\end{goodsit}
\begin{proof}
    By Lemma \ref{lem:ttsq=0FFcond}, $s+\justq_1+\justq_2=0$ remains true (meaning that no small/quarter item gets packed into an empty bin). Suppose an item of size more than $\niceLB$ remains unpacked. At this point $e$ has been reduced to 0 and the rule of last resort fails. The unpacked item is at most large by Invariant \ref{inv:threats}, which holds by Lemma \ref{lem:ttsq=0FFcond}.
    No bins are empty in the end, and all bins are in $\calC\cup\Qonebig\cup\Delta$. We first consider the weight according to $\wb$. No bin can be in $\Dnice\cup\Dlarge$ at this stage. (If the bin $\Dnice$ received an item of weight at least 2 without entering $\Dtwonice$, we consider it to be in $\calC$.)
    The only bin that can have $\wb$-weight less than 4 is $\Dhalf$, which has weight 2. The only way that an item that is at most large remains unpacked is if the item is in fact large and the bin is $\Dsh$. Moreover, if there exists a bin with $\wb$-weight 6 (for example, a bin in $\nice$)
    we are in Good Situation \ref{gs:weight}. 

    If the bin $\Dhalf^+$ does not exist, it cannot be created once $c\ge\FFcond$ and $s+\justq_1=0$ by Lemma \ref{lem:ttsq=0FFcond}: $s+\justq_1$ will remain 0 and small and quarter items are packed into $\calC$ first, not into $\justDhalf$. 
    Thus when the preconditions of this good situation held for the first time, $\Dhalf^+$ already existed. From this point on (when $e=\Omega(m)$), the only items that are packed into empty bins are nice, large or dominant.

    If a dominant item arrives after this point, then if no bin with $\wb$-weight at least 6 is formed and $s+\justq_1=0$, it is packed into an empty bin. We must have $\Dnice=\emptyset$ at this time.    Moreover, as this bin is preceded by at least $\FFcond$ bins in $\calC$, it will remain alone in its bin. In this case no nice item can arrive anymore (and none arrived earlier either). We now consider the weight according to $\wn$ and note that the packed $\wn$-weight is at least $4m-2$ as $\Dlarge$ has weight 4. Then the unpacked item is not large, a contradiction. 

    We next consider arriving nice items.
    If there is no bin with $\wb$-weight 6, then $n=0$ and no nice item gets packed into $\calC$. Packing a nice item into $\justQ2\cup\justQ1$ does not make it complete, while packing a nice item into $\Delta$ violates the nice rule, so no nice item is packed in Step 1.
    There are at most three nice items in $\Dnice\cup\Dtwonice$. Apart from these bins, any nice items not packed in Step 1 must get packed using the rule of last resort into $\Delta$ (because $\calC$ is already tried in Step 1 and $\nice=\emptyset$), so there are $O(1)$ nice items in the whole input.   
    Since $e=\Omega(m)$ when we entered this good situation, $\Dsh$ already existed and no dominant item is packed into an empty bin, $\Omega(m)$ of the empty bins must receive pairs of large items if an item remains unpacked.

We again find a contradiction using $\wn$, as we have 8 $\wn$-weight in $\Omega(m)$ bins, compensating for the deficit in $\Dtwonice$ and $\Dsh$.
\end{proof}

\begin{theorem}
\label{thm:firstphaseworks}
    Throughout the execution of the starting phase, as long as we are not in a good situation the algorithm maintains Invariant \ref{inv:threats} even when using the rule of last resort, and the algorithm does not fail to pack any item. Furthermore, we have $e=\Omega(m)$ or $\justq_2>0$ throughout. Whenever $s+\justq_1>0$, we have $r\le\NEzero-1$.
\end{theorem}
\begin{proof}
    By Lemma \ref{lem:tte>0} and Lemma \ref{lem:ttsq=0FFcond} we maintain Invariant \ref{inv:threats} as long as the rule of last resort is not applied, so in particular as long as there is an empty bin, or if $s+\justq_1=0$ and $c\ge\FFcond$ before a bin enters $\topblock$ by the rule of last resort. 
    
    As long as we are not in a good situation but we are in the starting phase, $\timing<\Tfirst(\countconstant)$. By Theorem \ref{thm:emptywhensq>0}, we then have $r\le\NEzero-1$ and $e\ge\Ezero+1$ if $s+\justq_1>0$. In this case we are done by Lemma \ref{lem:tte>0}, as the rule of last resort was not yet applied. 

    Suppose $s+\justq_1=0$.
    Whenever $\justq_2>0$, we are done by Theorem \ref{thm:q2weightbound} unless the rule of last resort gets applied. In that case however, if $\justq_2>0$ when $e$ becomes 0, at this point already $\wq(I)=4m$ by Theorem \ref{thm:q2weightbound}, because all bins are in $\topblock\cup\Dnice$ if $s=\justq_1=e=0$. This is Good Situation \ref{gs:weight}. If $\justq_2=0$ when $e$ becomes 0, then a bin in $\justQ2$ could only be created later if $\justq_1>0$ and $e=0$ at the same time, which contradicts Theorem \ref{thm:emptywhensq>0}.

    Suppose $s+\justq_1+\justq_2
    =0$. We now consider the last time when $s+\justq_1>0$ was true (or the start of the input if it was never true). Recall that at this time we had $e\ge\Ezero+1 $ and $r\le\NEzero-1$. The rule of last resort had not yet been applied, and we are done if $c\ge\FFcond$ by Good Situation \ref{gs:wb}.   Suppose $c<\FFcond$.
    From this point on only items that are at least nice were packed into empty bins. When $s+\justq_1+\justq_2=0$, all regular bins are in $\calC\cup\Delta$, so $e\ge m-\FFcond-4=\Omega(m)$ so Invariant \ref{inv:threats} holds by Lemma \ref{lem:tte>0}.

    Clearly, the algorithm does not fail to pack an item as long as $e>0$, as an empty bin is an option for all items. Suppose we reach the state $e=0$ by packing an item in the last empty bin. Then $\justq_2>0$ by the above and $s+\justq_1=0$ by Theorem \ref{thm:emptywhensq>0}. Then $\topblock\cup\Dnice=\justQ2\cup\calC\cup\nice\cup\Delta=\relev\cup\wcomplete$ so by Theorem \ref{thm:q2weightbound}, a $\wq$-weight of $4m$ has already arrived (since $e=0$), and we are in Good Situation \ref{gs:weight}.
\end{proof}

\section{The fill-up phase}
\label{sec:secondphase}

\subsection{Preliminaries}
\label{sec:prelim}
Once the fill-up phase is reached we refer to the bins by their type they had when the fill-up phase was started and no longer update these sets. E.g., a bin in $\justQ2$ at the start of the fill-up phase that receives a big item in the fill-up phase does {not} become a bin in $\Qtwofive$ but is referred to as a bin in $\justQ2$ even after receiving the big item. 
We assume all bins in $\S$ contain more than $\ffsmall$, overestimating its total content by at most $2\cdot(\ffsmall)$. The bin in $\Dlarge$ is assumed to contain a large$^+$ item once we enter this phase, overestimating its content by at most $3+\eps$. The bin $\Dsh$ contains an \easytwo item which matches the packing rules in the fill-up phase. 

There are three possible states when entering the fill-up phase:
\begin{itemize}
    \item $s+\justq_1>0$
    \item $s+\justq_1+\justq_2=0$
    \item $s+\justq_1=0$ and $\justq_2>0$
\end{itemize}
The first case is what we will call the standard case where $r\le\NEzero$ and $e\ge \Ezero$ holds by Theorem \ref{thm:emptywhensq>0} (after possibly one item was packed into a previously empty bin), for which we will see that when using our packing rules we will eventually end in a good situation or the input ends. For the second and third case we will see that we are already in good situations that do not have requirements on $e$ or $r$ (Good Situations \ref{gs:nonempty} and \ref{gs:q2} in section \ref{sec:secondeasy}). From section \ref{sec:firstthree} we focus on the standard case.
\begin{definition} 
    Let $e_0$ be the number of empty bins at the start of the fill-up phase. 
\end{definition}
For the fill-up phase, we introduce the set $\cal U$ of \emph{unused} bins. These are mostly bins that have not received items in the fill-up phase but that we do plan to use for items. At the start of the fill-up phase, these are the bins that are not in $\wcomplete\cup\Qmatch$,  so $|\mathcal{U}| =m-\lc-\qmatch=\tilde m=r+e_0$. 

Bins in $\Dthreenice$ are also not used for items anymore, but are initially counted as part of $\cal U$ so that $\largethreat\le u$ (see Invariant \ref{inv:large}).
Bins in $\Qfive$ are initially in $\cal U$ to maintain the proper ratio $q_1:q_2$ (Theorem \ref{thm:q1q2}). 
At the start of the fill-up phase, the bins in $\cal U$ are sorted from left to right. We use the ordering\footnote{The at most two bins in $\Dnice\cup\Dtwonice$ can be placed anywhere. However, they are ignored when determining $\beta_0$ later.}
\[
\Dthreenice,\Qtwofive,\justQ2,\S,\Q1,
\calE,\justDhalf\cup\Dlarge
\] 
where the subsets $\S$ and $\Q1$ are ordered by non-increasing levels and $\Dsh$ is placed among them if it exists. Indeed, the entire set $\cal U$ is essentially sorted by the levels of small and quarter items at the start of the fill-up phase, so for instance bins in $\Q1$ (including bins in $\Qonefive$ and $\Qonebig$) have level at most $\ffsmall$ for the sorting. 
Throughout the fill-up phase, by the level of a bin in $\calQ$ we will always mean the total size of the quarter items in this bin at the end of the starting phase. Regarding $\nice$, it is often convenient to divide the contents of these bins in a part of size at least $\smallslot$ and a part of size exactly $\bigLB$ (and this is why these bins are first in the ordering). See the definition of $\Tfirst(k)$ and the proof of Lemma \ref{lem:beta}.

During the fill-up phase, we will maintain a set $\calD$ such that $\topthreat\le u-d$ will hold throughout the fill-up phase. 
We define a specific initial set $\calD$ below and we will update this set throughout, using the following rules.

\begin{description}
    \item[Rule 1] 
    \label{rule:avg}Each bin that is used (in particular bins in $\calD$) will receive at least $\bigLB$ (including parts assigned to a bin but not packed in it) to ensure $\bigthreat\le m-\bigblock$ continues to hold. We already note that for bins that are {empty} at the start of the fill-up phase the bound of $\bigLB$ can be reached simply by using Next Fit (it will hold for all but at most one bin at any time).
    \item[Rule 2]  
    \label{rule:d}Whenever some item cannot be packed into some bin in $\calD$ that already received items of the same type in the fill-up phase (types are defined below), that bin will leave $\calD$ and $\cal U$. Each time we pack and/or assign $\topLB$ to $\calD$ in the fill-up phase, a new bin is added to $\calD$. (Sometimes we will assign parts of items packed into other bins to bins in $\calD$.)
    \item[Rule 3] 
    \label{rule:nond}Each bin that is not in $\calD$ will receive at least $\topLB$ on average to maintain $\topthreat\le u-d$.
\end{description}

\paragraph*{Reducing the unused bins} 
We begin the fill-up phase by removing the bins in $\Dthreenice$ and $\Qfive$ from the unused bins. The contents of these bins were and remain counted.  
For some later calculations it will still be important that these bins may exist, which is why we include them initially and gave a specific ordering for them. 

Recall that the initial value of $u$ is $m-\lc-\qmatch=\tilde m$. By Theorem \ref{thm:finalphaseworks},
$\topthreat\le \initd(m-c-\countconstant)$ (because we put $\Qonebig$ into $\Qonefive$).
So at least 
\begin{align*}
m-\initd(m-c-\countconstant)&=\dsize m+\initd\cdot(\countconstant+c)
\\ &=\dsize (m-c) + \initd\cdot\countconstant +c
\end{align*}
bins will not receive top items in the fill-up phase (but $c$ of those bins are unavailable for top items anyway). 
Then after removing $\Dthreenice$ and $\Qfive$ from the unused bins, we have $u=m-c$, so at least 
\[
\dsize\cdot u+\initd\cdot\countconstant 
\] 
\emph{unused} bins will not receive top items. 
We will initially set $d=\dsize\cdot u+\initd\cdot\countconstant$. 

Analogously, $\bigthreat\le m-c-\countconstant$, so at least $c+\countconstant$ bins will not receive big items, meaning that at least $\countconstant$ unused bins will not receive big items in the fill-up phase.
While packing the input, at any time there will be \emph{half-full} bins. These are bins which have received some items during the fill-up phase but have not yet received (or been counted for) $9-\eps$ or, in the case of non-$\calD$ bins, $10+6\eps$. These half-full bins need to be taken into account to ensure that Invariant \ref{inv:threats} is maintained. Denote their number by $h$.

\begin{invariant}
\label{inv:d}
    At any time, the number of bins in $\calD$ that have not yet received any item in the fill-up phase it at least
    $\dsize\cdot u+\initd\cdot\countconstant-h$, where $u=m-c$ initially and $u$ is updated according to Rule \hyperref[rule:d]{2}.
\end{invariant}

\begin{lemma}
\label{lem:inv-d}
    As long as we pack items according to  
Rule \hyperref[rule:avg]{1} and update $\calD$ and $\cal U$ according to Rule \hyperref[rule:d]{2} for all except at most 10 bins, or pack items according to Rule \hyperref[rule:nond]{3},
and at most 13 bins are half-full at any time,
$\topthreat\le u-d$, $\bigthreat\le m-\bigblock$ and Invariant \ref{inv:d} are maintained.
\end{lemma}
\begin{proof}
    If we pack items according to Rule \hyperref[rule:nond]{3}, the claims follow from the fact that we pack at least $\topLB$ in every non-$\calD$ bin, decreasing $\topthreat$ and $\bigthreat$ by at least $1$ while increasing $\bigblock$ by at most 1 and removing exactly $1$ bin from $\cal U$ when we start using a new bin. Hence $u-d$ and $\topthreat$ both decrease by 1, and $\bigthreat$ decreases by at least 1. In this case the ratio $d:u$ increases.

    Regarding items that get packed according to Rule \hyperref[rule:avg]{1}, we pack at least $\bigLB$ in every bin in $\calD$ (and then remove such bins from $\calD$ and $\cal U$) and add a new bin to $\calD$ after packing $\topLB$, which means that we add a bin on average after using $\frac{\topLB}{\bigLB}$ bins in $\calD$. The ratio $d:u$ remains constant during this process apart from at most one bin.

    The ratio can be seen as follows. After packing a total size of $x$ into $\calD$, we have removed at most $x/(\bigLB)$ bins from $\calD$ (because we only remove a bin from $\calD$ once we start using the next one) and we have added $\lfloor x/(\topLB) \rfloor$ bins to $\calD$.
    Ignoring the rounding, overall $d$ has decreased by at most $x(\frac{1}{\bigLB}-\frac1{\topLB})$
    and $u$ has decreased by at most $x/(\bigLB)$. The ratio is maintained. The rounding means that the set $\calD$ may be 1 smaller during processing. This together with the initial value $d=\dsize\cdot u+\initd\cdot\countconstant$ leaves 10 bins for which the rules do not need to be followed, since $\initd\cdot\countconstant>11$. 

    Finally, maintaining $\bigthreat\le m-\bigblock$ given that initially $\bigthreat\le m-c-\countconstant$ means that it is sufficient (due to the bin $\Dtwonice$) that at most 13 bins will be half-full at any time during the fill-up phase. 
\end{proof}
 
A consequence of Invariant \ref{inv:d} is that the set $\calD$ does not become empty during the packing if we indeed maintain $h\le 13$.
Maintaining this invariant means that all remaining big and top items can be packed at any point during the fill-up phase. Note that if at some point indeed very many top items arrive, they can perhaps not all of them be packed outside of $\calD$, as there can be various half-full bins outside of $\calD$. However, by Invariant \ref{inv:d}, as long as the total number of half-full bins is at most 13, all top items can indeed be packed.

\paragraph*{Item types}
Naturally, as we start filling up bins in the fill-up phase, new thresholds become important. Rather than leaving enough space for items that may arrive in the future as in the starting phase, we now want to use the remaining space efficiently. Bins in $\S$ have at least $(\onlbin)-(\smallslot)=\domslot$ space remaining and bins in $\justQ2$ leave at least $(\onlbin)-\qtwoUB=\topLB$ space. 
If some item type fits at least three times on top of bins in $\justQ2$, then First Fit gives a stronger bound on the packing and Rule \hyperref[rule:avg]{1} is satisfied. (Lemma~\ref{lem:rule13} formalizes this.) As there is at least $\topLB$ remaining space, we define small items to be of size at most $\smallUBF$. Items of size more than $\heavyUBF$ that fit twice in this space are also small items. Apart from the range $(\niceLB,5+3\eps]$, these size ranges are a subset of what defined small items in the starting phase.

The next items are \emph{quarter} items which may not fit three times on $\justQ2$ but at least three times on $\S$. For these we need to be slightly more careful; this is described below.
\emph{Small} items fit at least four times in an empty bin and at least two times on $\S\cup\justQ2$. It can be seen that two items that are larger than small items satisfy Rule \hyperref[rule:avg]{1}. To fill the remaining space well, the remaining items are split into quarter$^+$ and quarter$^{++}$ items.

Analogously, \emph{\easytwo} items fit twice on $\S$ and satisfy Rule \hyperref[rule:nond]{3}, explaining the (unchanged) threshold $\halfUBF$. These items have good sizes as well as large weights.
To pack large items in the range $(\halfUBF,\largepUBF]$ efficiently we introduce a new threshold at $\largemUBF$ separating \emph{large$^-$} and \emph{large$^+$} items which will be explained later. 

\begin{figure}[ht]
	\label{fig:phase2types}
	\centering
	\includegraphics[width=\textwidth]{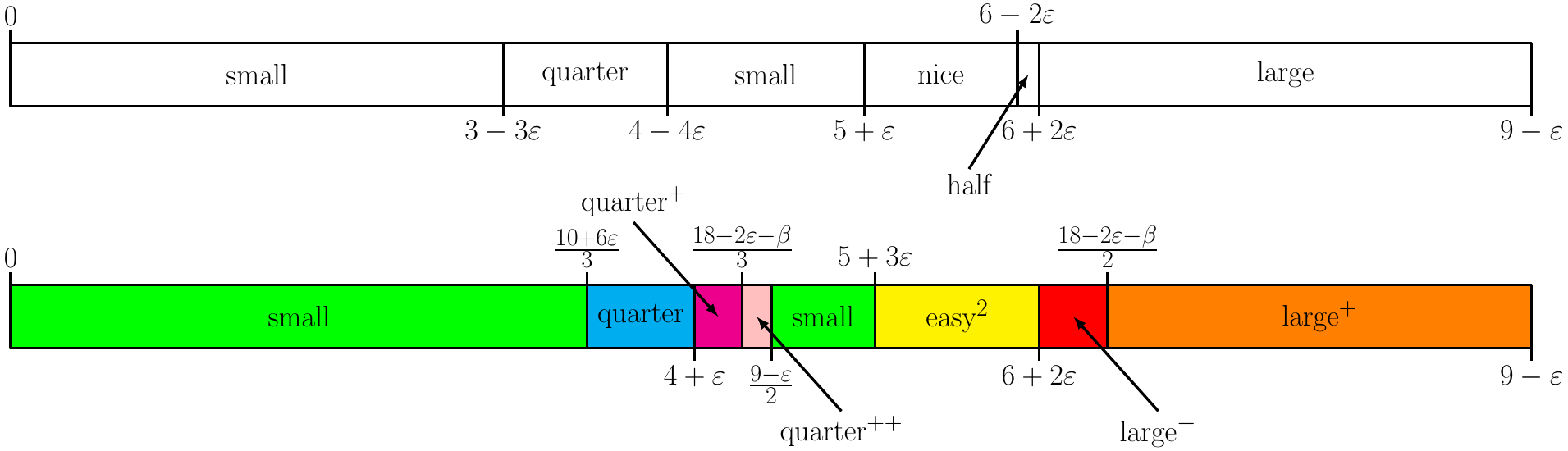}
	\caption{(Sketch, using $\eps=1/31, \beta=5$) A comparison of the important thresholds for our algorithm in the starting phase (top) and in the fill-up phase (bottom). The small (green) items fit at least three times (resp. twice) in the remaining empty space in a $\justQ2$ bin ($10+6\eps$), the easy$^2$ (yellow) items fit at least twice in the remaining empty space in a $\S$ bin ($12+4\eps$) and the large$^-$ (red) (resp. quarter$^+$ (purple)) items fit at least twice (resp. three times) in the remaining empty space in a bin filled to at most $\beta$ ($\onlbin-\beta$).}
    \label{fig:phase2typescap}
\end{figure}
The value $\beta$ will be defined later and will change during the fill-up phase; we will have $\beta\in(\smallUB,\smallslot]$. There are ten size ranges, but the items in six ranges are straightforward to pack (see below). We call the quarter$^+$, quarter$^{++}$, large$^-$ and large$^+$ items \emph{hard} items.

\medskip
\noindent
\begin{tabular}{cccc}
\textbf{Type}              & \textbf{Max size}      
\\ \hline \rule{0pt}{2.6ex}
small                  & $\smallUBF$           
\\
quarter                & $\quarterUBF$              
\\
quarter$^+$              & $\frac{\onlbin-\beta}{3}$ 
\\
quarter$^{++}$              & $\heavyUBF$               
\\
small & $\bigtwoUBF$             
\end{tabular}$\quad$
\begin{tabular}{cccc}
\textbf{Type}              & \textbf{Max size}\\ 
\hline \rule{0pt}{2.6ex}
\easytwo                 & $\halfUBF$               
\\
large$^-$              & $\largemUBF$              
\\
large$^+$              & $\largepUBF$              
\\
big & $\topLB$ 
\\
top & \offbin 
\end{tabular}

\medskip 
We see that quarter items and small items may be slightly larger than in the starting phase. This does not decrease their weight. From the starting phase, we will only use the facts that bins in $\justQ2$ have level at most $8-8\eps$ and bins in $\S$ have level at most $\smallslot$; the old size thresholds play no further part in the analysis.

\subsection{Packing methods for non-hard items}
\label{sec:secondphasetypes}

\begin{description}
    \item[small and big] These are items such that when packing them into $\justQ2$ using First Fit (which are the fullest bins (ignoring $\X$) that possibly still need to receive items in the fill-up phase), each bin (apart from constantly many) will receive at least $\bigLB$ of items. 
    
    These items are always packed in the leftmost available bin that did not already receive items from another type in the fill-up phase. That is, we essentially use First Fit, but items of the three different size ranges are not packed together into bins. We only use bins in $\calD$. This is feasible because new bins will enter the set $\calD$ as we pack items in it (Rule \hyperref[rule:d]{2}).

    \item[\easytwo and top] When packing \easytwo or top items using First Fit into any bin that is not in $\justQ2\cup\Dlarge\cup\Dnice$, Rule \hyperref[rule:nond]{3} is automatically satisfied. (We have $\abs{\Dlarge\cup\Dnice}\le2$.)

    These items are packed exactly like the small items except that we skip the bins in $\justQ2$. We use First Fit for \easytwo and top items (on the bins not in $\justQ2$) separately. 
    \item[quarter] The quarter items are packed by First Fit, alternating between bins in $\justQ2$ and bins not in $\justQ2$. 
    If we run out of bins in $\justQ2$ we use First Fit on dedicated bins not in $\justQ2$. 
    If we run out of nonempty bins not in $\justQ2$ we can always pack all remaining items as we will show in Good Situation \ref{gs:q2}.

    \item[hard] These are the quarter$^+$, quarter$^{++}$, large$^-$ and large$^+$ items. These are the only types of items that we will sometimes (have to) pack into empty bins although nonempty bins that are not in $\justQ2$ are still available.
    When packing hard items of one type using First Fit into empty bins, Rule \hyperref[rule:nond]{3} is satisfied \emph{and} there is a surplus above the requirement of $\topLB$. We use this surplus to pack as many items as possible into nonempty bins, which is necessary to succeed. See Section \ref{sec:hard}.
\end{description}
Once First Fit uses a new bin for an item, the bin previously considered (in which the item did not fit anymore) is removed from $\cal U.$

\subsection{Hard (quarter$^{+(+)}$ and large) items}
\label{sec:hard}
The above lemmas show that in the fill-up phase, all items except quarter$^{+(+)}$ and large items can be packed while following the rules---as long of course as quarter$^{+(+)}$ and large items are packed following the rules as well. We next consider quarter$^{+(+)}$ and large items separately. For this we first need to define some important sets.

\subsubsection{Sets considered for packing hard items}

\paragraph*{The set $\S_{-L}$ and the parameter $\beta_0$}

At the beginning of the fill-up phase, if there are $e$ empty bins, then if we pack two large items in each empty bin, there will be $e$ nonempty bins that will \emph{not} be needed to pack large items that have not arrived yet, even if $u$ of them arrive in total. In principle, we let $\S_{-L}$ be the $e$ \emph{rightmost} nonempty bins. We only deviate from this if there are less than $e$ nonempty bins. In that case $\S_{-L}$ consists of all nonempty unused bins (after the removal of $\Dthreenice\cup\Qtwofive$). 
We define $\beta_0$ as the level of the leftmost bin in $\S_{-L}$ unless all other nonempty bins are in $\Dthreenice\cup\Qtwofive$, in which case $\beta_0$ is set to $\smallslot$. Of course, due to other items arriving, it may well be that some large items end up getting packed in $\S_{-L}$ after all.  

Since $u=m-\lc-\qmatch$ at the start of the fill-up phase, it is not possible that more than $u-\qfive$ large items arrive in the fill-up phase. 
This holds because each bin counted in $\qfive+\qmatch$ contains two items larger than $\bindiff$ or a dominant item, and each bin counted in $\ell$ contains two items larger than $\bindiff$ or an item larger than $\offhalf$.
Moreover, as above,
Invariant \ref{inv:d} is maintained by Theorem \ref{thm:finalphaseworks}. As soon as $u$ starts to decrease in the fill-up phase, more than $u-q_5$ large items may still arrive.

\paragraph*{Construction of the set $\calD$}
We initially define $\calD$ as the {leftmost} $\dsize\cdot u+\initd\cdot\countconstant$ unused bins (where bins in $\Qfive$ and $\Dthreenice$ have already been eliminated).

\paragraph*{Definition of the set $\S_L$}
We define $\S_L$ as the set of the remaining nonempty bins in $\cal U$ (that is, all the unused bins that are not in $\calD\cup\S_{-L}\cup\calE$). 
The set $\S_L$ may be empty.

\medskip 
Note that the above construction is done only once. The sets $\S_L$ and $\S_{-L}$ do not increase and a bin leaves such a set only if the bin leaves $\cal U$ or is added to $\calD$. While packing items in the fill-up phase, bins will be added to $\calD$ from left to right.
Hence entering $\calD$ will happen first to bins in $\S_L$ and then (possibly) to $\S_{-L}$ and $\calE$. Such bins leave their original sets.

Bins in $\S$ are filled to at most $\smallslot$ and hence have at least $\domslot$ empty space. 
Easy$^2$ items fit at least twice.
Bins in $\S_{-L}$ are filled to at most $\beta_0$ and hence have at least $\onlbin-\beta_0$ empty space. 
Large$^-$ items therefore fit at least twice (explaining that threshold). 
Bins initially in $\calD$ are loaded to at most $\qtwoUB$ and hence have at least $\onlbin-\qtwoUB=\topLB$ space.

\subsubsection{Packing methods for hard items: Five stages}
Hard items are packed as follows. There are five possible stages; depending on what the packing looks like when the fill-up phase starts, not all stages might be applicable. Generally, we start by taking advantage of any $\Qonefive$ bins that are available, then we start filling the nonempty bins from right to left, always trying to avoid using empty bins as much as possible.

Hard items are always distributed among several types of bins ($\calD,\S_L,\S_{-L},\calE$). One of the used types will always be $\calD$. 
In several cases, we will specify that nonintegral numbers of bins need to be used for some items. This can be implemented as follows. We always start by using a bin that is not in $\calD$. Then, we keep using bins in $\calD$ until we would exceed the desired ratio. At that point we use a bin that is not in $\calD$ again, and repeat the process. In the long term we get closer and closer to the desired ratio.

\paragraph*{Stage 1} In this stage, as mentioned above we first exploit any bins in $\Qonefive$ that may exist. These bins can be intermixed with $\justQ1$ in the sorted order, but we use them first. To be more precise, these bins are not used to pack any new items (as they are already quite full) but rather to pack additional items into $\calD$, as follows. Note that these bins are not in $\cal U$. For this stage we define the upper bound for quarter$^+$ items as $\heavyUBF$ and the upper bound for large$^-$ items as $\bigLB$ and hence neither quarter$^{++}$ items nor large$^+$ items exist. This also means that the value of $\beta$ (see table of item type thresholds for the fill-up phase) does not yet play a role.

\begin{description}
    \item[Quarter$^{+(+)}$:] As long as there exist bins in $\Qonefive$ with partially uncounted contents, for packing quarter$^{+(+)}$ items such bins are considered to contain quarter$^{+(+)}$ items (of the size of the ignored items). Each such bin allows us to pack $\frac{\topmbig}{1-3\eps}$ bins in $\calD$ with two quarter$^{+(+)}$ items each. In these $\frac{\topmbig}{1-3\eps}$ bins we then pack (or count) more than $2\frac{\topmbig}{1-3\eps}(\quarterUBF)+(\topmbig)=\frac{\topmbig}{1-3\eps}\cdot(\bigLB)$, which includes the so far uncounted part of the bin in $\Qonefive$. 
    \item[Large:] As long as there exist bins in $\Qonefive$ with partially uncounted contents, for packing large items such bins are considered to contain large items (of the size of the ignored items). Each such bin allows us to pack $\frac{\topmbig}{\smallUB}$ bins in $\calD$ with one large item each. In these $\frac{\topmbig}{\smallUB}$ bins we then pack (or count) more than $\frac{\topmbig}{\smallUB}(\largeLB)+(\topmbig)=\frac{\topmbig}{\smallUB}\cdot(\bigLB)$, which includes the so far uncounted part of the bin in $\Qonefive$. 
\end{description}

It can be seen that by packing quarter$^{+(+)}$ and large items this way, Rule \hyperref[rule:avg]{1} and Rule \hyperref[rule:nond]{3} are followed.
Once we run out of bins in $\Qonefive$, we start using the methods in the following table. The value of $\beta$ will change over time and is set at the start of each of the following stages. 
Note here that for $\beta\leq\heavyUBF$ the upper bound for quarter$^+$ items is at least  the upper bound for quarter$^{++}$ items and hence quarter$^{++}$ items do not exist as long as $\beta\leq\heavyUBF$.

\begin{center}
    
$
\begin{array}{cccccccc}
\label{tbl:heavylarge}
\text{Item type} & \text{Bin type} & \text{Nr bins} & \text{Per bin} & \text{Counted} & \text{Average} 
\\
\hline \rule{0pt}{2.6ex}
\text{quarter$^+$} & \calD 
 & 3 & 2 & 2(\quarterUBF) & 9+\frac{9}{4}\eps
\\
& \justQ1,\S_{-L}(,\S_L),\calE & 1 & 3 & 3(\quarterUBF) & & 
\\[0.6ex]
\hline\rule{0pt}{2.6ex}
\text{quarter$^{++}$} 
& \calE & 1 & 4 & \frac43 (\onlbin-\beta) & \bigLB & 
\\
& \calD 
 & \frac{45-4\beta-5\eps}{2\beta-(9-\eps)} & 2 & \frac23 (\onlbin-\beta) &  
\\[0.6ex]
\hline\rule{0pt}{2.6ex}
\text{large$^-$ }  & \calD & 1 & 1 & \largeLB & 9+3\eps &  
\\
 & \justQ1,\S_{-L}(,\S_L),\calE   & 1 & 2 &  \domslot  & 
\\[0.6ex]
\hline\rule{0pt}{2.6ex}

\text{large$^+$ }  & \calD & \frac{\onlbin}{\beta}-2 & 1 & \frac{\onlbin-\beta}2 & \bigLBf 
\\
 & \calE  & 1 & 2 &  \onlbin-\beta   & 
\end{array}
$
\end{center}

The third column indicates the number of bins used each time. The column `Average' contains the average amount counted over all bin types used for this item. 

This table is implemented as follows. Items from these four types are packed into separate bins. For each type except quarter$^{++}$ items, we first use a bin in $\calD$. For quarter$^{++}$ items,
we use $\calE$ first to ensure that the average packed in the bins with these items is always greater than $\bigLB$, apart from at most one bin: the bin currently being used. For all other types, we have this guarantee for all used bins up to and including the most recently filled bin in $\calE$, which is all but at most four bins for each type.

For each bin used we pack items in it until we have packed the number of items in the column Per bin. 
We first use bins in $\calD$ until we have packed the number of items in the column Per bin and until we have used the number of bins in the column Nr bins, followed by one bin in $\calE$. 
(If the number of bins supposed to be used in $\calD$ is not an integer, then the number of bins used in $\calD$ is always off by less than 1 compared to the desired ratio of $\calD:\calE$ usage.) 
This procedure keeps repeating. Whenever we start using a new bin, the previous bin for this type is called \emph{closed}. Once a bin is closed it is removed from $\cal U$. 

In the following we will always 
check if bins with level at most $\beta$ exist. While this is true we remain in the current stage. Else we update $\beta$ to the next higher bound.  

\paragraph*{Stage 2}  
We set $\beta:=\min(\beta_0,\ffsmall)$.
As long as there are bins with level at most $\beta$ available, we use these bins (including bins in $\S_L$ if needed) for packing quarter$^{+}$ and large$^-$ items. 
There are no quarter$^{++}$ items and only large$^+$ items are packed into empty bins. 

\paragraph*{Stage 3} As Stage 2, but with $\beta:=\min(\beta_0,\frac{9-\eps}2)$. Bins in $\S_L$ with level at most $\beta$ will still be used for quarter$^{+}$ and large$^-$ items. There are no quarter$^{++}$ items and only large$^+$ items are packed into empty bins. 

\paragraph*{Stage 4} 
This stage only exists if $\beta_0>\frac{9-\eps}2$. We set $\beta:=\beta_0$.

Quarter$^{++}$ and large$^+$ items are packed into empty bins. All items are packed according to the table. Bins in $\S_L$ are \emph{not} used for quarter$^+$ items; we would go to Stage 5 instead. This is done to ensure that enough bins are left for large$^-$ items and that we do not waste bins on quarter$^+$ items.

\paragraph*{Stage 5} Only bins in $\S_{L}\cup\calD\cup \calE$ are left. At this point quarter$^+$ and large$^-$ items are also packed into empty bins.

\medskip 
Stage 6 would be the stage where all empty bins have been filled while some bins in $\S_L\setminus\calD$ remain. The case where all nonempty bins get filled first is discussed below. We will show that Stage 6 is not reached or we are in a good situation.

As is in the starting phase, if some item cannot be packed according to the packing rules (including the case where we change the packing rules if we reach a good situation) we use the rule of last resort.

\section{Analysis of the fill-up phase}
\label{sec:finish}

To analyze the fill-up phase, it will be convenient to discuss how many bins of various types there are as a fraction of $\tilde m$, which is $m-\lc-\qmatch$ at the start of the fill-up phase (before removing any bins). We define $e_0'$, for instance, to be the number of empty bins at the start of the fill-up phase \emph{divided by $\tilde m$}. We do the same for all other bin types that we refer to and that are in $\cal U$ at the start of the fill-up phase.

\begin{lemma}
    \label{lem:rule13}
    In every stage and for every item type, Rule \hyperref[rule:avg]{1}, Rule \hyperref[rule:d]{2} and Rule \hyperref[rule:nond]{3} are followed (if applicable) except for at most 13 bins.
\end{lemma}
\begin{proof}
    Rule \hyperref[rule:d]{2} is followed automatically. 
    \begin{description}
        \item[Small and big items]      Items in the range $(\bigLB,\topLB]$ fit once in each bin in $\calD$ and items in the range $(\heavyUBF,\bigtwoUBF]$ fit twice. 
    This proves Rule \hyperref[rule:avg]{1} is followed  for these items. At most one of these bins has received only one item in the range $(\heavyUBF,\bigtwoUBF]$ at any time.
    
    Bins in $\calD$ are currently counted for $\ffsmall$ if they are in $\justQ2$ (after reassigning some contents to $\Q1$). Nothing above $\smallslot$ in $\justQ2$ has been counted so far. 
    Consider a bin in $\justQ2$ that is used to pack items in the range $(0,\smallUBF]$. At least three of these items fit in each bin. 
    Consider a bin that contains $x\in(\smallslot,\qtwoUB]$ at the start of the fill-up phase. 
    By Lemma \ref{lem:FFvarbin}, apart from at most 6 bins, such a bin ends up being full to at least $x+\frac{3}4(\onlbin-x)\ge\frac{3}{4}(\onlbin)+\frac{1}{4}(\smallslot)=15-3\eps$ 
and we can count it in the fill-up phase for $15-3\eps-(\smallslot)>\bigLB$.

Now consider a bin in $\calD\setminus\justQ2$ that contains $x\le\smallslot$. Apart from at most 6 bins, such a bin contains at least $x+\frac34(\onlbin-x)$ in the end but has only been counted for $x$. We can therefore count at least $\frac34(\onlbin-x)>\bigLB$
in the fill-up phase.
\item[Easy and top items] Follows immediately from the use of First Fit and the size thresholds. All bins used for packing \easytwo and top items contain (and are counted for) at most $\smallslot$ at the start of the fill-up phase. At most one bin is half-full at any time.
\item[Quarter items]
    From the use of First Fit and the size threshold it follows immediately that Rule \hyperref[rule:nond]{3} is followed. A bin in $\calD$ that cannot receive another quarter item is filled to at least $(\onlbin)-(\quarterUBF)=\quarterffull$. Such a bin is counted for at most $\smallslot$ so we count it for at least $(\quarterffull)-(\smallslot)=8+3\eps$ in the fill-up phase. As we use one bin that can be counted for $\topLB$ for each such bin, we achieve an average of $((8+3\eps)+(\topLB))/2>\bigLB.$
    \end{description}

    For hard items, see the comments in Stage 1 and the table above (the column Average corresponds to Rule \hyperref[rule:nond]{3}). Each of the four hard item types adds at most one half-full bin.

    In total there are at most 6 bins for which we do not follow Rule \hyperref[rule:avg]{1} and at most $1+6+1+1+4=13$ bins are half-full at any time.
\end{proof}

\subsection{Easy cases: bins in $\S_L\cup\S_{-L}$ get filled up early}
\label{sec:secondeasy}

\begin{goodsit}[Very Simple Fill-Up]
\label{gs:nonempty}
    Only empty bins, half-full bins and bins in $\calD$ are left unused before reaching Stage 6.
\end{goodsit}
\begin{proof}
    As long as there are bins in both $\calD$ and $\calE$ we keep packing items as usual. Suppose only empty and half-full bins remain unused.
    By Lemma \ref{lem:rule13}, we followed Rules \hyperref[rule:avg]{1}--\hyperref[rule:nond]{3} up until this point, so Invariant \ref{inv:d} holds by Lemma \ref{lem:inv-d}. Due to Rule \hyperref[rule:avg]{1} and the volume bound given by Theorem \ref{thm:finalphaseworks}, less than a total volume of $(9-\eps)e$ can still arrive. By Lemma \ref{lem:inv-d}, the algorithm does not fail on dominant items. Each such item takes away one bin and more than $9-\eps$ of size that can still arrive. Other items fit at least twice in empty bins, meaning that these bins will certainly become more than $9-\eps$ full on average if we just use First Fit.

    Now suppose only bins in $\calD$ and half-full bins remain unused. This means we have already packed $\topLB$ in (or can assign $\topLB$ to) every bin that is not in $\calD$  (Rule \hyperref[rule:nond]{3}) apart from the half-full bins, meaning that all items have been packed. 
\end{proof}

\begin{goodsit}[Weight-based packing with $\justQ2$]
\label{gs:q2}
    We have $\justq_2>0$, there is at least one bin in $\justQ2$ that did not receive any item in the fill-up phase, and all remaining bins are half-full or empty.
\end{goodsit}

\begin{proof}

    Recall that $\Qtwofive$ is removed from $\cal U$ at the start of the fill-up phase, so only bins in $\Q2\setminus\Qtwofive=\justQ2$ may remain in $\cal U$.
    If at some point all bins are empty or half-full, we are in Good Situation \ref{gs:nonempty}.
    Otherwise, we show that for the input $I$ packed so far we have $\wf(I)=4m$ if the assumption of this lemma holds, so we are done by Good Situation \ref{gs:weight}. 
    
In this situation, some bins in $\justQ2$ have not yet received items in the fill-up phase, so all items that are packed in $\cal D$ are in fact packed into bins in $\justQ2$ and add weight to the bins (which we do not need in most cases). 
We reassign weight to the bins in $\justQ2$ as described in Theorem \ref{thm:q2weightbound}. Since this already ensures that these bins have a weight of 4, we can ignore the easy-sized items as they get packed exclusively in $\justQ2$ in the fill-up phase in this case. This also includes small and large items packed in Stage 1. We have $\qonebig=0$.

We next consider the types that are not packed into $\justQ2$ at all. Those are the \easytwo items. Items in the range $(5+3\eps,6+2\eps]$ get packed in pairs and add at least $4$ weight to a bin; at most one bin contains only one such item at any time (2 weight missing). An item in the range $(10+6\eps,12]$ adds $4$ weight on its own. 

Now we consider the types that are packed into bins in $\justQ2$ and bins not in $\justQ2.$ For all these types, the number of bins used in $\calD$ per bin used outside of $\calD$ is at least 1.
For both types of large items we pack at least two items with weight 2 each into the non-$\justQ2$ bins, and for both types there is at most one bin at any time with weight 2 and not 4. A bin in $\justQ2$ always receives such an item first, compensating for the at most one bin (per large type) that has one large item. 

For quarter$^{++}$ items that are packed in Stage 4 or 5 we pack at least four items with weight 1 each into the non-$\justQ2$ bins. Hence for all these types we also add 4 weight per non-$\justQ2$ bin. At most one non-$\justQ2$ bin with quarter$^{++}$ items contains less than four of them (but at least one): 3 missing weight. 

We are left with quarter items and quarter$^+$ items (and quarter$^{++}$ items that are packed as quarter$^+$ items; we refer to them as quarter$^+$ items in the remainder of this paragraph). Those require some reassignment. For each non-$\justQ2$ bin that receives quarter$^+$ items we pack six quarter$^+$ items into bins in $\justQ2$, we reassign the weight of all of these items to the non-$\justQ2$ bin. Hence all non-$\justQ2$ bins with quarter$^+$ items are assigned at least 6 weight, even if such a bin contains less than three quarter$^+$ items (there can be one such bin). 

For each non-$\justQ2$ bin that receives (at least) three quarter items we pack (at least) two quarter items into bins in $\justQ2$. We reassign the weight of these items to the non-$\justQ2$ bin. At most one non-$\justQ2$ bin has less than three quarter items but at least one quarter item. Overall each non-$\justQ2$ bin (except at most 1 which receives 3) received at least 4 additional weight in the fill-up phase due to quarter 
items: 1 missing weight.

In $\Delta$, at most 6 weight is missing: 2 weight per bin, but not in $\Dhalf^+$ as that bin is counted among bins with \easytwo items; if on the other hand $\justDhalf$ exists, then $\Dlarge$ does not exist.
In total, we see that compared to the requirement of 4, a weight of at most $2+3+1+6=12$ is missing in bins outside of $\justQ2$, and at most $2q_2$ weight in the bins in $\justQ2.$

By the assumptions of this lemma, every nonempty bin that is not in $\justQ2$ received items in the fill-up phase. So, by assigning one weight from the at least $\qcminusthree+2q_2$ bins in $\Q1$
and using the result of Theorem \ref{thm:q2weightbound}, all bins end up with weight at least 4. (If $\Dhalf^Q$ exists, one less weight is needed.)
By Good Situation \ref{gs:weight}, all items will be packed.
\end{proof}

\subsection{The first three stages succeed}
\label{sec:firstthree}
At the start of the fill-up phase, we have the following situation. First of all, if there is no bin in $\S\cup\justQ1$, by (\ref{eq:relev}) and Observation \ref{obs:deltabounds} all except at most four bins that are not in $\wcomplete\cup\Qfive\cup\Qmatch\cup\Dthreenice$ are in $\calE\cup\justQ2$. These cases were  handled in Good Situation \ref{gs:nonempty} and Good Situation \ref{gs:q2} and they do not rely on any of the properties in the following list except $\abs{q_1-2q_2}=O(1)$ (which holds by Theorem~\ref{thm:q1q2}). This includes the cases where we enter the fill-up phase with $s+\justq_1+\justq_2=0$ (Good Situation \ref{gs:nonempty}) or $s+\justq_1=0$ and $\justq_2>0$ (Good Situation \ref{gs:q2}). 
In the remainder of the proof we consider the case $s+\justq_1>0$.

We have the following properties. Recall that $\eps=1/31$.

\begin{itemize}
      \item There exists a bin in $\S\cup\justQ1$.
      \item We have $r\le\NEzero$, $\tilde m=\Omega(m)$ and $e\ge \frac{1-5\eps}{\ffsmall}\tilde m+\frac12\cnice+\frac{1-5\eps}{\ffsmall}\qmatch+\eps \lc-\NEzerocplusone=\Ezero=\Omega(m)$ 
      by Theorem \ref{thm:emptywhensq>0}.
      \item All bins in $\wcomplete$ (except at most $\xsize$) have (assigned) level at least $\offbin.$ Since $\timing<\Tfirst(\countconstant)$ held until the current item arrived, this means that at most $(\ENDBIG)(m-\lc-\qmatch)+(\bigLB)(q_5+\cnice)+O(1)$ is counted on the regular bins. 
    \item 
    In $\calQ$ there is more than $(1+7\eps)\qonefive$ ignored volume, since $(\topLB)-(\bigLB)=1+7\eps$. 
    
    \item We have $\abs{q_1-2q_2}=O(1)$ by Theorem~\ref{thm:q1q2}.
    
    \item There can be at most $\xsize$ bins in $\X$ which we will ignore for most of the analysis.
\end{itemize}

\begin{lemma}
    \label{lem:beta}
    If $e_0'<\frac27$, 
    then 
    \begin{equation}
        \label{eq:beta}
        \beta_0\le
\frac{ 2(\ENDBIG) - 2e_0'(\smallslot)  - 2\cnice' (\nicenobig) }{2 - 5e_0'-2\cnice'}.
    \end{equation}
\end{lemma}

\begin{proof}
    Consider the bin that determines the value of $\beta_0$. Suppose it is in $\Dthreenice$. Then the regular $e_0$ bins to its immediate right are counted for at least $(\ffsmall)e_0$, and the $\tilde m-2e_0$ bins to its left are counted for more than $\nicetime$. Given that $(\ENDBIG)m+O(1)$ is counted
    \footnote{Here we ignore $\bigLB$ that is in each bin in $\Dthreenice\cup\Qfive$ as the threshold $\Tfirst$ is exactly higher by $(\bigLB)(\cnice+q_5)$ and what we additionally count in the bins is higher than at least $(\bigLB)(\cnice+q_5)$.}, this implies $e_0'>\frac38(1-\eps)>\frac13$, a contradiction.

    Now suppose that bin is in $\Q2$. This means the number of bins in $\Q1$ is at most $e_0$, so the number of bins in $\Q2$ is at most $e_0/2$ and the number of bins in $\Dthreenice$ is at least $\tilde m-\frac52e_0$.  
    These bins contain and are counted for more than $\nicefill$ (three nice items). From each bin we can assign $\nicefill-(\bigLB)-(\ENDBIG)=3+3\eps>\ENDBIG$ to empty bins\footnote{This motivates the lower bound on the size of nice items.}. As long as $\tilde m-\frac52e_0\ge e_0$, this gives an average level over all the bins of more than $\ENDBIG$, a contradiction. Since $e_0<\frac27\tilde m$, this case cannot occur.
        
    If the bin is in $\S\cup\Q1$, then $\beta_0$ as a function of $e_0'$ is maximized if the number of bins in $\calQ$ is maximized 
    and $\justq_2>0$ (due to our ordering of the bins). In particular, if all the $e_0$ bins in $\S_{-L}$ contain only a quarter item (except for the one that defines $\beta_0$).
    In this case there are (at least) $e_0$ bins in $\S_{-L}$ that contain more than $\smallUB$ and $e_0/2$ bins in $\calD\cup\S_L$ that contain more than $\smallslot$, counting only the (initial) quarter items in both cases. The remaining regular bins contain at least $\beta_0$, and $\cnice$ of them contain at least $\nicefill$. Using Theorem \ref{thm:sjustq1g0} and the definition of $\Tfirst$,
    we get the equation 
\[
(\smallslot)\frac{e_0}{2}+(3-3\eps)e_0+ (\tilde m-\frac52e_0-\cnice)\beta_0 +(\nicefill) \cnice\leq (\ENDBIG)\tilde m+(\bigLB)\cnice
\]
    which implies (\ref{eq:beta}).
\end{proof}

Of course we also have $\beta_0\le\smallslot$. The bound (\ref{eq:beta}) is stronger than this for $e_0'\le\frac27$. In this region, the bound is weakest if $\cnice'=0$, since it is decreasing in $\cnice'$.

\begin{corollary}
If $e_0'<\frac27$, then 
\[
e_0'\ge
    \frac{2 (\beta_0-(\ENDBIG)+(\nicenobig-\beta_0)\cnice')}{5 \beta_0- 12 (1 - \eps)}.
\]
\end{corollary}

\begin{corollary}
\label{cor:beta}
    If $\beta_0>\ffsmall$ and $e_0'<\frac27$, then $e_0'>(2-10\eps)/(8-8\eps)$.
    If $\beta_0>\heavyUBF$, then $e_0'>(6-6\eps)/(21+19\eps)=18/67=0.26866$ 
    (for $\eps=1/31)$.
\end{corollary}

\begin{lemma}
\label{lem:13}
    We do not run out of empty bins in Stages 1--3.
\end{lemma}
\begin{proof}
    In Stage 1, nothing is packed into empty bins. In Stage 2, when packing large$^+$ items into $\calE$, we pack $\frac{18-2\eps}{\beta}$ large items for every empty bin that we use, where $\beta\le\ffsmall$. Thus we run out of large items after packing $\frac{\beta}{18-2\eps}\tilde m\le\frac{30}{139}\tilde m<\frac{13}{60}\tilde m=\Ezero'\tilde m+O(1)$ bins.\footnote{This inequality is the one that motivates setting $\eps=1/31$. We have equality for $\eps=2\sqrt{\frac{19}3}-5=0.03322$. (But if we used that value of $\eps$, the algorithm would only work if we had an unbounded number of machines.)} Hence some empty bins are still empty at this point. No other items are packed into empty bins in Stage 2.

    In Stage 3 we have $\beta\le\heavyUBF$. The number of large items packed when all empty bins are full is at least $\frac{\onlbin}{\beta}e_0'$ since no small items are packed into empty bins until after Stage 3. Thus if $e_0'>1/4$, all large items are packed before running out of empty bins, meaning that we do not run out of empty bins since no other items are packed there.

    Suppose $e_0'\le1/4$.
Using Lemma \ref{lem:beta} and the fact that $\beta_0\le\smallslot$, we have
$\beta_0\le (\largeLB-2e_0'(\smallslot))/(2-5e_0')$. This number is minimized if $\beta_0$ is maximized (as a function of $e_0'$). We conclude that we pack at least
\[
    \ell^*=\frac{(18-2\eps)(2-5e_0')}
    {\largeLB-2e_0'(\smallslot)} e_0'
\]

large items if we use all empty bins for them. This is at least $1$ for $e_0'\in[\Ezero',1/4]$, and $e_0'$ cannot be smaller. This means that we do not run out of empty bins for $e_0'\le1/4$ either.
\end{proof}

\subsection{All items are packed}
\label{sec:allitemspacked}
In the next part, we will be formulating several linear programs. We introduce a common notation. It will be convenient to scale all the numbers so that $\tilde m=1$. Whenever we speak of a scaled number, we mean the corresponding number of bins divided by $\tilde m$.

\noindent 
\begin{tabularx}{\textwidth}{cX}
$x_1$ (resp.,~$x_2$) & the scaled number of bins in $\Qonefive$ that are removed from $\cal U$ in Stage 1 while packing large (resp., small) items.\\
$y_1$ & the scaled number of empty bins that receive two large$^+$ items before Stage 5.\\
$y_5$ & the scaled number of empty bins that receive two large items in Stage 5.\\
$x_3$  (resp.,~$x_4$) & the scaled number of bins in $\justQ1\cup\S_{-L}$ that receive two large$^-$ items (resp., small items).\\
$x$ & the total scaled number of nonempty bins with level at most $\heavyUBF$.
\end{tabularx}

\begin{lemma}
\label{lem:lowbins}
    The scaled number of nonempty bins with level at most $\heavyUBF$ is at least $e_0'$ if $\beta_0\le\heavyUBF$,
    and it is at least $\frac{90-278e_0'}{57}$. For $e_0'>18/67$, the first bound is the strongest (but it only applies if $\beta_0\le\heavyUBF$) and for $e_0'\le18/67$, the second bound is strongest.
\end{lemma}
\begin{proof}
    If $\beta_0\le \heavyUBF$ all $e_0'$ bins in the set $\S_{-L}$ have level at most $\heavyUBF$ (some of them can be in $\Qonefive$ and contain additional items).

    We reassign parts of quarter items from $\Q2$ to $\Q1$ so that all bins in $\calQ$ have the same level, which is more than $\ffsmall$.
    Let the scaled number of bins with level more than $\heavyUBF$ after this reassignment be $a$. This includes all bins in $\Dthreenice$. 
    Suppose it includes no bins in $\calQ$.
    There are $1-e_0'-a$ bins with level more than $\ffsmall$ but at most $\heavyUBF$. 
    We have \[
    (\nicefill)\cnice'+(1-e_0'-a)(4-4\eps)+(a-\cnice')\heavyUBF\le\ENDBIG+(\bigLB)\cnice',\] so
\[
    a\le \frac{120e_0'-26}{19}.
\]
We see that at most $q_2'+\frac{120e_0'-26}{19}$ bins have actual level more than $\heavyUBF$ in this case. We have $q_2'+q_1'\le 1-e_0'-a$, so $q_2'\le(1-e_0'-a)/3$ (by (\ref{eq:q1inv})). The scaled number of bins with actual level more than $\heavyUBF$ is
\[
q_2'+a\le\frac{1-e_0'}3+\frac23a\le \frac{1-e_0'}3+\frac{240e_0'-52}{57} = \frac{221e_0'-33}{57}.
\]
Therefore the scaled number of nonempty bins with level at most $\heavyUBF$ is at least $1-e_0'-\frac{221e_0'-33}{57} = \frac{90-278e_0'}{57}$.

Finally, we consider the case where all bins in $\calQ$ end up with level more than $\heavyUBF$ after the reassignment. We get the same bound on $a$ as above but this now includes the set $\Q2$. So the upper bound on $a$ immediately is an upper bound for the number of bins with actual level more than $\heavyUBF$. This is a weaker bound than in the first case, implying a stronger bound on the scaled number of nonempty bins with actual level at most $\heavyUBF.$
\end{proof}

\begin{lemma}
\label{lem:betasmall}
    If $\beta_0\le\heavyUBF$, we do not run out of empty bins in Stage 5.
\end{lemma}
\begin{proof}
    If $\beta_0\le\frac{9-\eps}2$, the algorithm first does Stage 3, after which all bins in $\S_{-L}$ are filled up by definition of $\beta_0$. Furthermore, all quarter$^{++}$ items are packed as quarter$^+$ items throughout the complete execution. By Lemma \ref{lem:13}, some bins are still empty when $\S_{-L}$ is filled. In fact, some bins in $\S_L$ are still empty at this point since otherwise we would be done by Good Situation \ref{gs:nonempty} or \ref{gs:q2}. This means that only quarter$^{+}$ and large items get packed on the $x$ bins with level at most $\heavyUBF$.
    We then go to Stage 5.

We first consider the case $e_0'\le18/67$ (constraint $c4$).
We have $x_1+x_2+x_3+x_4\ge \frac{90-278e_0'}{57}$ by Lemma \ref{lem:lowbins} (constraint $c5$). 
The scaled number of \emph{nonempty} bins packed before Stage 5 is at least 
$2x_1+4x_2+2y_1$, and this is at most $1-e_0'$ (constraint $c1$). Here we have that the coefficient of $x_1$ is $1+\frac{\topmbig}{\smallUB}$ as we use $\frac{\topmbig}{\smallUB}$ bins in $\calD$ for every bin in $\Qonefive$ in Stage 1 when packing large items. With $\eps=1/31$ we have $1+\frac{\topmbig}{\smallUB}>1.4222$ and we use a weaker constraint in our presented linear program. Analogously, the coefficient of $x_2$ is $1+\frac{\topmbig}{\smallUB}$ as we use $\frac{\topmbig}{1-3\eps}>1.357$ bins in $\calD$ for every bin in $\Qonefive$ in Stage 1 when packing small items.
Before Stage 5 we use $1$ (resp. $3$) bins in $\calD$ for every nonempty bin outside $\calD$ when packing large$^-$ (resp. small) items, explaining the coefficient $2$ (resp. $4$) of $x_3$ (resp. $x_4$) in constraint $c1$. For every empty bin used to pack large$^+$ items we use $\frac{\onlbin}{\beta}-2$ bins in $\calD$. Hence, the coefficient of $y_1$ is $\frac{\onlbin}{\beta}-2\geq 2$ (for $\beta\leq\beta_0\leq\heavyUBF)$.

We next count items that are at least large and bins containing pairs of half items together. Their scaled sum is at least $3x_1+4y_1+3y_5$. This value is at most 1 (constraint $c2$). Here the coefficient of $x_4$ is 1 because only the bins in $\Qonefive$ when packing small items contain two half items or an item that is at least large.
The coefficient of $x_3$ is 3 as we pack three large items per group of bins used for large$^-$ items. For the large$^+$ items we use $\frac{\onlbin}{\beta}-2\geq 2$ bin in $\calD$ which receive one large item and 1 empty bin which receives two large items, so in total 4 large items get packed per group used for large$^+$ items (explaining the coefficient of $y_1$). The coefficient of $y_5$ is the minimum of the coefficients of $x_3$ and $y_1.$
    
The total scaled number of bins packed when running out of empty bins is at least $1.4222x1 + 2.357x2 + 2x3 + 4x4 + 3y1 + 2y5 + 4(e0-y1-y5)$. The value is at most 1 (constraint $c3$). Here the coefficients are as in constraint $c1$, except that for $y_1$ it is 3 as we now also count the empty bin used per group when packing large$^+$ items. As in constraint $c2$ the coefficient of $y_5$ is the minimum of the coefficients of $x_3$ and $y_1$ in constraint $c3.$ Finally, the remaining empty bins are filled by quarter$^+$ items and such group consists of three $\calD$ bins and one empty bin, so 4 bins in total (explaining the coefficient of $e_0-y_1-y_5$).

The total scaled number of empty bins that received two large items can be at most $e_0'$ (constraint $c6$).

We summarize our above constraints in a form of a linear program, which can be solved  using the GLPK solver\footnote{\url{https://cocoto.github.io/glpk-online/}}, for example. This linear program has no feasible solution.

\begin{verbatim}
var e0 >= 0;
var x1 >= 0;
var x2 >= 0;
var x3 >= 0;
var x4 >= 0;
var y1 >= 0;
var y5 >= 0;
maximize
obj: e0;
subject to
/* nonempty bins packed early */
c1: 1.4222*x1 + 2.357*x2 + 2*x3 + 4*x4 +2*y1 <= 1-e0;
/* large items packed */
c2: 1.4222*x1 + x2 + 3*x3 + 4*y1 + 3*y5 <= 1;
/* number of bins packed total */
c3: 1.4222*x1 + 2.357*x2 + 2*x3 + 4*x4 + 3*y1 + 2*y5 + 4*(e0-y1-y5) <= 1; 
/* to ensure that (90-278*e0)/57 > e0 */
c4: e0 <= 18/67; 
/* number of used bins with low level */
c5: x1 + x2 + x3 + x4 >= (90-278*e0)/57;
/* pairs of large items in empty bins */
c6: y1 + y5 <= e0;
\end{verbatim}

We next consider the case $e_0'>18/67$ and adapt condition c4 accordingly. We replace condition c5 by $x_1+x_2+x_3+x_4\ge e_0'$. Again there is no feasible solution.
\end{proof}

Using Corollary \ref{cor:beta} together with Lemma \ref{lem:13} and Lemma \ref{lem:betasmall}, we are left with the case $\beta_0>\frac{\bigLB}{2}$ and $e_0'>18/67=0.2686$. 

\begin{lemma}
    \label{lem:nofailstage5}
    If $\beta_0>\heavyUBF$ and we do not run out of empty bins during Stage 4, all items get packed in Stage 5.
\end{lemma}
\begin{proof}
    If we do not run out of empty bins in Stage 4, all inequalities in the last linear program above still hold, except that the coefficient of $y_1$ decreases, because the lower bound for the ratio $\frac{\onlbin}{\beta}-2$ (from table \eqref{tbl:heavylarge}) decreases to at least $1$ (for $\beta\leq (\smallslot)$ compared to $\beta\leq\heavyUBF$). We further note that the ratio $\frac{695-62\beta}{31\beta-139}$ (from the table on page \pageref{tbl:heavylarge} with $\eps=1/31$) for quarter$^{++}$ items is larger than the ratio $3$ for quarter$^+$ items.
    Overall, we get the following linear program that still has no solution and conclude that all items get packed. From now on, we will  present only the constraints for the linear programs.
\begin{verbatim}
c1: 1.4222*x1 + 2.357*x2 + 2*x3 + 4*x4 + y1 <= 1-e0; 
c2: 1.4222*x1 + x2 + 3*x3 + 3*y1 + 3*y5 <= 1;
c3: 1.4222*x1 + 2.357*x2 + 2*x3 + 4*x4 + 2*y1 + 2*y5 + 4*(e0-y1-y5) <= 1;
c4: e0 >= 18/67;
c5: x1 + x2 + x3 + x4 >= e0;
c6: y1 + y5 <= e0;
\end{verbatim}
\end{proof}

We are left with the case where $\beta_0>\heavyUBF$, $e_0'>18/67$ and we run out of empty bins in Stage 4, meaning that only quarter$^{++}$ items and large$^+$ items are packed into empty bins. 

\begin{lemma}
\label{lem:lowbins2}
    We have 
    \[
        x\ge \frac{188-278e_0'+62\beta_0(2e_0'-1)}{82-31\beta_0}.
    \]
\end{lemma}
\begin{proof}
All bins counted in $x$ that have level at most $\ffsmall$ are in $\Q1$. Suppose there are $\tilde x$ such bins. This means that there are $\tilde x/2$ bins in $\Q2$. 
For the calculations, we assign the part of the contents of those bins above $\beta_0$ to the $\tilde x$ bins in $\Q1$. From each bin we reassign $6-6\eps-\beta_0$, and this gets distributed over two bins with level at least $3-3\eps$.
Then those bins reach level at least $6-6\eps-\frac{\beta_0}2$. 
The remaining $x-\tilde x$ bins have level at least $\ffsmall$.

We have $1-2e_0'$ bins with level at least $\beta_0$, $x-\tilde x$ bins with level at least $\ffsmall$ and $\tilde x$ bins with level at least $6-6\eps-\frac{\beta_0}2$ and $e_0'-x$ bins with level at least $(\bigLB)/2$. In order to maximize $1-e_0'-x$, the level on all bins should be as low as possible given the constraints, since the total volume is fixed to be $3+\eps$. Thus in the worst case $\tilde x=x$ and  all levels are equal to the lower bounds just given. We get the claimed bound.
\end{proof}

\begin{corollary}
    We have $\beta_0\le \frac{188-360e_0'}{62-155e_0'}$.
\end{corollary}
\begin{proof}
    For larger values of $\beta_0$, we get that $x>e_0'$ which is impossible. 
\end{proof}

\begin{lemma}
\label{lem:nofailstage4}
    We do not run out of empty bins during Stage 4.
\end{lemma}
\begin{proof}
    The regular coefficients for $y_1$ and $e_0-y_1$ depend on $\beta_0$ and decrease in $\beta_0$. On the other hand the lower bound for $x_1+x_2+x_3+x_4$ depends on $\beta_0$ but is increasing in $\beta_0$ (for $\beta_0$ in $[(\bigLB)/2,\smallslot]$ and $\eps=1/31$).

Hence we give two intervals for $\beta_0$ and linear programs with worst possible coefficients and bounds.

If $\smallslot\geq\beta>5.2$ we have the following linear program that does not have a solution. Here we have that the coefficient of $y_1$ is $\frac{\onlbin}{\beta}-2\geq\frac{\onlbin}{\smallslot}-2>1.088$ (for $\eps=1/31$). The coefficient of $e_0-y_1$ is the number of bins used when packing quarter$^{++}$ item, so $\frac{45-4\beta-5\eps}{2\beta-(\bigLB)}+1\geq\frac{45-4(\smallslot)-5\eps}{2(\smallslot)-(\bigLB)}+1>9.17.$ We note that quarter$^+$ items are not packed into empty bins before Stage 5.
\begin{verbatim}
c1: 1.4222*x1 + 2.357*x2 + 2*x3 + 4*x4 + 1.088*y1 <= 1-e0;
c2: 1.4222*x1 + x2 + 3*x3 + 3.088*y1 <= 1;
c3: 1.4222*x1 + 2.357*x2 + 2*x3 + 4*x4 + 2.088*y1 + 9.17*(e0-y1) <= 1; 
c4: e0 >= 18/67;
c5: x1 + x2 + x3 + x4 >= (188-278*e0+62*5.2*(2*e0-1))/(82-31*5.2);
c6: y1 <= e0;
\end{verbatim}

If $5.2\geq\beta>(\bigLB)/2$ we have the following linear program that does not have a solution. Here we have that the coefficient of $y_1$ is $\frac{\onlbin}{\beta}-2\geq\frac{\onlbin}{5.2}-2>1.449$ (for $\eps=1/31$). The coefficient of $e_0-y_1$ is $\frac{45-4\beta-5\eps}{2\beta-(\bigLB)}+1\geq\frac{45-4\cdot5.2-5\eps}{2\cdot5.2-(\bigLB)}+1>17.78.$
\begin{verbatim}
c1: 1.4222*x1 + 2.357*x2 + 2*x3 + 4*x4 + 1.449*y1 <= 1-e0;
c2: 1.4222*x1 + x2 + 3*x3 + 3.449*y1 <= 1; 
c3: 1.4222*x1 + 2.357*x2 + 2*x3 + 4*x4 + 2.449*y1 + 17.78*(e0-y1) <= 1; 
c4: e0 >= 18/67;
c5: x1 + x2 + x3 + x4 >= 
    (188-278*e0+62*(9-1/31)/2*(2*e0-1))/(82-31*(9-1/31)/2);
c6: y1 <= e0;
\end{verbatim}
\end{proof}

\begin{corollary}
    We do not run out of empty bins in Stages 1--5, so we never reach Stage 6.
\end{corollary}

We conclude that therefore all items end up packed (by Good Situation \ref{gs:nonempty} and $\ref{gs:q2})$.

\section{Dependency on $\eps$  and $m$}
\label{sec:depm}

The approach used by our algorithm cannot work in this way for $\eps\ge1/5$, because we might run out of empty bins in the starting phase: we only have the bound $r\le\NEzero$, which may be as much as $\frac{\ENDBIG}{\ffsmall}m+\NEzeroc$. Thus the starting phase cannot always be completed successfully for $\eps=1/5$ (which corresponds to a competitive ratio of 1.46667). 

Moreover, already for much smaller values of $\eps$ we quickly run out of empty bins in the fill-up phase if only large items arrive (at first). See the proof of Lemma \ref{lem:13} and its footnote.
This means that a number of additional nonempty bins may receive a single large item. It may still be possible to pack all remaining items after this, since they are all of size at most 6, but a different algorithm (which has a Stage 6) would be required.

In Lemma~\ref{lem:13} we need for $\beta\leq \ffsmall$ that 
$\frac{\beta}{18-2\eps}\tilde m<\Ezero'\tilde m$. Hence, it holds that there is at least one empty bin if
\begin{align*}
\frac{\ffsmall}{\onlbin}\tilde m+1&\leq \left(\left(\frac{1-5\eps}{\ffsmall}\tilde m + \frac{1-5\eps}{\ffsmall}\qmatch+ \frac{4\eps}{\ffsmall}\ell-\NEzerocplusone\right)/\tilde m\right)\tilde m\\&\leq \frac{1-5\eps}{\ffsmall}\tilde m+\frac{\nicerzero}{\ffsmall}\cnice + \frac{1-5\eps}{\ffsmall}\qmatch+ \frac{4\eps}{\ffsmall}\ell-\NEzerocplusone.
\end{align*} 
This leads to \begin{align*}
    \frac{\ffsmall}{\onlbin}(m-\qmatch-\lc)\leq
    \frac{1-5\eps}{\ffsmall}(m-\qmatch-\lc)+ \frac{1-5\eps}{\ffsmall}\qmatch+\frac{4\eps}{\ffsmall}\lc-\NEzerocplustwo
\end{align*}
    and with $\eps=1/31$ to 
\begin{align*}
    \left(\frac{13}{60}-\frac{30}{139}\right) m+\frac{30}{139}\qmatch+\frac{271}{8340}\lc=\frac{7}{8340}m+\frac{30}{139}\qmatch+\frac{271}{8340}\lc \geq\NEzerocplustwo.
\end{align*}
As $\qmatch\geq0$ and $\lc\geq0$ this holds if $m\geq\frac{\NEzerocplustwo\cdot8340}{7}$, which holds if $m\geq\mmin.$

\bibliographystyle{plain}
\bibliography{stretch}

\appendix 
\section{First Fit}
Our algorithm often uses First Fit (FF) as a subroutine, sometimes on bins (really parts of bins) that do not all have the same size. FF has the following useful properties. 

\begin{lemma}\footnote{The simple proof of this lemma can be found, e.g., in~\cite{53bin}.}
	\label{lem:2+bin}
	Consider a set $V$ of bins that is packed by First Fit
	of which at least the last $\abs{V}-2$ bins contain at least $k$ items.
	If $\abs{V}\ge3$, the total level of the bins in $V$ is more than 
	\[\frac{k\abs{V}}{k+1}\,.\]
\end{lemma}

\begin{lemma}
    \label{lem:FFvarbin}
    For any set of $v$ variable-sized bins that is packed using First Fit, the following property holds. If at least $k<v/2$ items are packed into each bin, the total size of all the items packed into these bins is at least 
    $$ \frac{k}{k+1}\sum_{j=k}^{v-k}s(j),$$
    where the size of the $j$-th bin is denoted by $s(j)$. This even holds if the number of bins increases while First Fit is running (in this case $v$ is the final number of bins).
\end{lemma}

\begin{proof}
    Let the bins be sorted by the order of First-Fit.
    
    We look at an $(k+1)$-tuple $(j,j+1,\ldots,j+k)$ with $1\leq j\leq v-k$. Let $\alpha$ be the largest empty space of bins $j,\ldots,j+(k-1)$. The items in bin $j+k$ have size at least $\alpha$. Bins $j,\ldots,j+(k-1)$ on the other hand are filled to at least $s(j)-\alpha,\ldots,s(j+(k-1))-\alpha$.
    We know that at least $k$ items of size at least $\alpha$ are packed in bin $j+k$, so in these $k+1$ bins we have an overall load of at least $\sum_{i=j}^{j+(k-1)}s(i)$.
    Applying this bound for $j=1,2,\dots$ we find guarantees for First-Fit of at least $\sum_{i=1}^{k} s(i) + \sum_{i=k+2}^{2k+1} s(i) +\ldots, \sum_{i=2}^{k+1} s(i) + \sum_{i=k+3}^{2k+2} s(i) +\ldots$, etc.
    Adding all these bounds gives
    \[(k+1)\cdot FF\geq \sum_{i=1}^{k-1} i\cdot s(i) + \sum_{i=k}^{v-k} k\cdot s(i) + \sum_{i=v-(k-1)}^{v-1} (v-i)s(i)>k\sum_{i=k}^{v-k}s(i).\qedhere \]
\end{proof}

\begin{corollary}
    \label{cor:FF11}
    For any set of $n>1$ bins that is packed using First Fit in which the $j$-th bin already contains $p(j)\ge p$, if at least one item is packed into each bin, the total size of all items packed into these bins (including the values $p(j)$) is at least $(\onlbin+p)(v-1)/2$.
\end{corollary}
\begin{proof}
    The assumptions imply that $k=1$ in Lemma \ref{lem:FFvarbin}. The total amount packed is therefore at least
    $\frac12\sum_{j=1}^{v-1}(\onlbin-p(j))+\sum_{j=1}^v p(j) > (\bigLBf)(v-1)+\frac12\sum_{j=1}^{v-1}p(j)>(\onlbin+p)(v-1)/2.$
\end{proof}

\subsection{Design constraints}\label{sec:ideas}

\bigskip
In this problem, it is very easy to make fatal mistakes already on very small inputs. As a simple example, just packing more than $\bindiff$ of items in one bin can already mean that some item remains unpacked at the end if $m$ items of size slightly less than 12 arrive. Such an input is feasible if $m$ is large enough and the small items are small enough to be packed with the items of size almost 12 into separate offline bins.

It would seem that items of size 3, for instance, could be packed into pairs, because if such items arrive it means that the input cannot contain $m$ items of size 12. However, packing just two bins with two such items each means that the algorithm fails if $m-1$ items of size exactly 12 arrive afterwards. On the other hand, packing three items of size $3$ into a single bin fails after $m$ items of size exactly $9$ arrive. Both of these are not an issue for $\eps=0$.

Being careful about putting too much into single bins is a design constraint for any bin stretching algorithm, regardless of the goal stretching factor $\stretchfactor$. This essentially forces any algorithm to start with a phase in which items are not packed above $\stretchfactor-1$ times the offline bin size unless there is a very good reason to do so (for instance, if items larger than $\stretchfactor-1$ arrive).
Our algorithm begins with a phase which does exactly this. However, the number of design constraints increases markedly once $\eps>0$.

Böhm et al.~\cite{bsvsv17onepointfive} considered the setting $\eps=0$, that is, a competitive ratio of $3/2$. In that setting, any item that fits twice in an optimal bin can be packed three times in an online bin. Moreover, crucially, such an item can be packed once in an online bin while still leaving an empty space of size at least an offline bin available, so an item of any size can still be packed there. If $\eps>0$, items in the range $(\bindiff,\offhalf]$ fit twice in an optimal bin, but cannot be packed in an online bin without blocking that bin for some (very large) items. Given that the combinatorial properties of items of size at most $\bindiff$ are different from those of items of size more than $\bindiff$ and the fact that offline bins have size $\offbin$, the threshold $\largeLB$ naturally becomes important as well. Items in the range $(\bindiff,\largeLB]$ are called \textbf{half} items. (Having a separate threshold at $\offhalf$, while possible, does not seem to be useful.)

Since half items block some items from being packed with them, it is important to pack half items efficiently, so (at least) twice per bin wherever possible; this motivates the bound of $\smallslot$ for packing small items, leaving space for two half items. As long as we pack only \textbf{small} items of size at most $\smallUB$ up to a level of at most $\smallslot$, using First Fit guarantees that we pack $\ffsmall$ per bin (apart from at most two bins).  

There are further complications when we consider items that are slightly smaller than half items, in the range $(\smallslot,\bindiff]$, a subset of what we call \textbf{nice} items.
Since these items can be larger than $\smallslot$, they cannot necessarily be combined with two half items in a bin. The easiest solution would be to pack these items in the same way as half items. However, an item in this range \emph{can} be packed together with an item larger than $\largeLB$ (a \textbf{large} item) in an offline bin, and two large items do not fit in $\domslot$ space. If we pack even a single large item (e.g., size $6+3\eps$) alone in a bin, and an item with size $6-3\eps$ together with some very small items ($m-1$ identical items of total size $2\eps$),
we would fail to pack the input if $m-1$ items of size $12-2\eps/(m-1)$ subsequently 
arrived. On the other hand, packing large items together with smaller items also does not seem good, as it could lead to filling many bins to (significantly) less than $\offbin$ and it is unclear how to compensate for that, for instance if many big items arrive later and bins exist that contain only 11 (for instance).

Finally, if $\eps=0$, items that fit four times in an offline bin can be packed twice in an online bin while still leaving empty space of size equal to an offline bin. Once we set $\eps>0$, items in the range $(\frac12\cdot(\bindiff),\frac14\cdot\offbin]$ fit four times in an offline bin, but cannot be packed in pairs without blocking bins for some  items.

Generally, once we set $\eps>0$ there is much more interaction between the packing rules for the several types than for $\eps=0$. For the specific inputs described before and similar inputs, it is at least relatively straightforward to identify sets of items that could cause trouble in the future as these usually consist of just one item repeated many times or perhaps two different items repeated many times. Items larger than half an online bin are packed one per bin by both the online and the offline algorithm and cause no problems if they arrive early in the input (unless we have already made one of the fatal mistakes described above of course). It is useful to think of such items mainly as a \emph{threat} which restricts how we can pack other items (see Definition \ref{def:threats} below).
The main threat at the beginning and until we have packed {$(2-6\eps)m$} is that $m$ items of size close to 12 will arrive.
There are also issues which only become apparent much later.

Packing items of size at most $\smallUB$ while packing at most $\smallslot$ into any bin using First Fit gives us a packing guarantee of $\ffsmall$ per bin (apart from at most two bins). One looming threat later in the input is $m$ items larger than 6 arriving. This can happen even until we have packed a total volume of nearly $6m$. Suppose that (as a simple example) we pack all the bins using First Fit, packing slightly more than $\ffsmall$ per bin, and then we start filling up the bins, packing one bin to above $11-3\eps$. Then the algorithm fails if $m$ items of size $7+\eps$ arrive at the end.

Generally, it is crucial to not wait too long before starting to fill up bins, as we might be forced to pack many large items alone into bins, meaning that many bins receive less than 12 (in the previous paragraph they get $11-3\eps$). This could easily lead to inputs that cannot be packed; we must strive to pack an average of at least 12 per bin when all is said and done.

We decide to start filling up bins after having packed essentially $(\ENDBIG)m$ using at most $\smallslot$ per bin. This will be the fill-up phase of our algorithm. Using the packing guarantee $\ffsmall$, this means that roughly $1/4$ of the bins are still empty at this point (all this will be made precise later). If also (or only) quarter items arrive, this can only be achieved by packing them carefully; just packing them one per bin gives only a packing guarantee of $3-3\eps$ (which means that we run out of empty bins too fast), while packing all of them  two per bin or even three per bin is impossible as mentioned above (the example with four items of size 3).

We will pack quarter items roughly in a 1-1-2 pattern: two quarter items alone in bins followed by a bin containing two quarter items. This pattern ensures that we still pack $\ffsmall$ per bin on average, and because in this way we pack four quarter items for every bin that we (possibly) pack above the level of $\bindiff$, all possible remaining top items in the input can always be packed. In fact our packing method will be slightly more conservative in creating new bins with two quarter items, depending on what other items arrive in the meantime. This is just one example of the interaction between packing rules for different items.

\end{document}